\documentclass[reqno,10pt]{amsart}
\usepackage{amssymb,amsmath,amsthm,amsfonts,color}

\usepackage{diffcoeff}

\usepackage{mathrsfs}

\usepackage{ulem}
\usepackage{enumerate}
\usepackage{threeparttable}

\usepackage{bm}
\usepackage[left=2.8cm,right=2.8cm,top=2.8cm,bottom=2.8cm]{geometry}
\usepackage{mathtools}
\usepackage{graphicx}
\usepackage[format=hang,labelfont=bf]{caption}

\usepackage[colorlinks=true, pdfstartview=FitV, linkcolor=blue,
            citecolor=blue, urlcolor=blue]{hyperref}

\allowdisplaybreaks
\numberwithin{equation}{section}
\theoremstyle{plain}
\newtheorem{theorem}{Theorem}[section]
\newtheorem{proposition}[theorem]{Proposition}
\newtheorem{lemma}[theorem]{Lemma}

\theoremstyle{definition}

\newtheorem{remark}[theorem]{Remark}

\newtheorem*{theorem*}{Theorem}

\begin{document}
\title{Discrete-velocity-direction models of BGK-type with minimum entropy: I. Basic idea}

\author{Qian Huang}
\address{Department of Energy and Power Engineering, Tsinghua University\\
    Beijing, 100084, China}
\email{huangqian@tsinghua.edu.cn}

\author{Yihong Chen}
\address{Department of Mathematical Sciences, Tsinghua University\\
    Beijing, 100084, China}
\email{chenyiho20@mails.tsinghua.edu.cn}

\author{Wen-An Yong*}
\address{Department of Mathematical Sciences, Tsinghua University\\
    Beijing, 100084, China}
\thanks{* Corresponding author}
\email{wayong@tsinghua.edu.cn}


\vskip .2truecm
\begin{abstract}
  In this series of works, we develop a discrete-velocity-direction model (DVDM) with collisions of BGK-type for simulating rarefied flows.
  Unlike the conventional kinetic models (both BGK and discrete-velocity models), the new model restricts the transport to finite fixed directions but leaves the transport speed to be a 1-D continuous variable.
  Analogous to the BGK equation, the discrete equilibriums of the model are determined by minimizing a discrete entropy.
  In this first paper, we introduce the DVDM and investigate its basic properties, including the existence of the discrete equilibriums and the $H$-theorem.
  We also show that the discrete equilibriums can be efficiently obtained by solving a convex optimization problem.
  The proposed model provides a new way in choosing discrete velocities for the computational practice of the conventional discrete-velocity methodology.
  It also facilitates a convenient multidimensional extension of the extended quadrature method of moments.
  We validate the model with numerical experiments for two benchmark problems at moderate computational costs.
\end{abstract}

\maketitle

\normalem

\section{Introduction}

Kinetic theories and the Boltzmann equation lay the foundation for investigating non-equilibrium many-body interacting systems. Besides the classical rarefied gas dynamics \cite{Shar2015}, kinetic theories have also found substantial applications in multiphase flow problems \cite{Fried2000,Huang2020} and the emerging field of active matter \cite{bel2021,Thu2014}.
However, solving the Boltzmann-like equation can be challenging. The first obstacle is due to the collision mechanisms of gas molecules \cite{Harris2004} or aggregation-breakage of aerosol/colloid particles \cite{Fried2000,MF2013}.
To overcome this obstacle, the BGK model \cite{bgk1954} and its improvements (for instance, Shakhov model \cite{shak1968} and ES model \cite{hol1966}) have been proposed to simplify the collision term while keeping fundamental properties of the original equation and hence become prevalent \cite{Harris2004}.

The second obstacle is the velocity dependence of the unknowns which makes solving the Boltzmann equation computationally costly. Many numerical approaches have thus been developed, including the linearization method \cite{Naris2005,Shar2015}, the spectral method \cite{dima2015}, the discrete-velocity method (DVM) \cite{ans2003,gat1975,Mieu2000,Mijcp2000}, the semi-continuous method \cite{kol2003,pre2003}, and a variety of moment methods
\cite{cai2015,lever1996,Chalons2017,MF2013,Huang2020}.
All these methods aim at good mathematical properties (realizability, model stability, convergence to the original equation, etc.) and numerical performance (computational complexity, numerical stability, etc.), and have their advantages and weaknesses \cite{bel2003}.

Among the methods above, the Gaussian-extended quadrature method of moments (Gaussian-EQMOM) \cite{Chalons2017,MF2013,yuan2012} serves as the original motivation of this project. In 1-D case, this method assumes the unknown or distribution to be a linear combination of $N$ ($N \geq 1$) Gaussian functions with unknown centers and a common (unknown) variance to be inversely solved from the velocity moments of the distribution \cite{yuan2012}. The resulting moment closure system was proved to satisfy the structural stability condition widely respected by physical systems \cite{HLY2020,Yong1999}.
The quadrature-based method of moments stands out also because of its potential to characterize systems far from equilibrium \cite{LM2022,MF2013}.
However, a satisfactory extension of EQMOM to 2/3-D velocity space(s) seems not available; see Refs.~\cite{Chalons2017,MF2013}.

On the other hand, we notice that a discrete-velocity-direction model (DVDM) was proposed physically and investigated in Refs.~\cite{Zhang2008,Zhang2013,Zhang2018}.
As a semi-continuous model different from the existing ones \cite{kol2003,pre2003}, the DVDM constrains particle transport in finite fixed directions while keeping the speed continuous. The resulting model has complicated terms for collisions and therefore does not offer a convenient way to get multidimensional EQMOM \cite{Zhang2008}.
Additionally, the semi-continuous system can be numerically solved by discretizing the transport speed in each orientation \cite{Zhang2008}, and thus provides a new way in selecting discrete velocity nodes, which is distinct from the common DVM practice in a uniform cubic lattice \cite{Mieu2000}.

In this project, we propose a discrete-velocity-direction model (DVDM) with collisions of BGK-type. For the new model, the equilibriums are determined by minimizing a discrete entropy. 
We investigate the existence and uniqueness of the discrete equilibriums and establish an $H$-theorem characterizing the dissipation property of the original kinetic equation. Because the DVDM is different from the discrete-velocity BGK models \cite{Mieu2000}, the analysis is quite involved in comparison with that in Ref.~\cite{Mieu2000}.
We also show that the discrete equilibriums can be obtained efficiently by solving a convex optimization problem.
Moreover, the continuous transport speed can be treated with simple discretizations or combinations with 1-D Gaussian-EQMOM.
The latter provides a new multidimensional extension of EQMOM and yields a hyperbolic moment system, which seems satisfactory.
The proposed methods are numerically validated with two benchmark tests. Further numerical and analytical results will be reported in forthcoming papers \cite{CHYinp}.

The remainder of the paper is organized as follows. Section \ref{Sec:dvdm} presents the discrete-velocity-direction model (DVDM) with BGK-collisions.
The existence of local equilibriums is studied in Section \ref{sec:dis_equil}.
Section \ref{Sec:solve_dvdm} provides two approaches to treat the continuous velocity-modulus in DVDM, including DVD-DVM in Section \ref{subsec:dvddvm} and DVD-EQMOM in Section \ref{subsec:dvd_eqmom}.
Section \ref{sec:numer} discusses numerical issues, with the space-time discretization schemes in Section \ref{subsec:sptm} and the computation of discrete equilibrium in Section \ref{sec:diseq_comp}.
In Section \ref{Sec:numsim}, two benchmark flows are simulated.
Finally, we conclude our paper in Section \ref{Sec:concl}.

\section{Discrete-velocity-direction models of BGK-type}
\label{Sec:dvdm}

Let $\bm \xi \in \mathbb{R}^{D}$ be the molecule velocity ($D=2,3$) and $\bm x \in \mathbb{R}^{D_x}$ the spatial position ($1 \le D_x \le D$). We consider the BGK equation \cite{bgk1954} for the distribution $f=f(t,\bm x, \bm \xi)$:
\begin{equation} \label{eq:bgk}
    \left \{
    \begin{aligned}
        \partial_t f + \bm \xi \cdot \nabla_{\bm x} f &= \frac{1}{\tau} \left( \mathcal{E}[f] - f \right), \quad t\ge 0, \\
        \mathcal{E}[f] &= \exp \left( \bm{\alpha}_{eq} \cdot \bm m (\bm \xi) \right), \\
        \bm{\alpha}_{eq} &= \left( \ln \frac{\rho}{(2\pi \theta)^{D/2}} - \frac{U^2}{2\theta}, \ \frac{\bm U}{\theta}, \ -\frac{1}{\theta} \right)^T \in \mathbb{R}^{D+2}, \\
        \rho \left( 1, \bm U, E \right)^T &= \langle \bm m(\bm \xi) f \rangle \in \mathbb{R}^{D+2}, \quad
        \theta = \frac{2E - U^2}{D}.
    \end{aligned}
    \right.
\end{equation}
Here $\tau$ is the relaxation time, $\bm m (\bm \xi) = \left( 1, \ \bm \xi, \ \frac{|\bm \xi|^2}{2} \right)^T \in \mathbb{R}^{D+2}$, $U=|\bm U |$ is the Euclidean length of the vector $\bm U$,
and the bracket $\langle \cdot \rangle$ is defined as $\langle g(\bm \xi) \rangle = \int_{\mathbb{R}^D} g(\bm \xi) d \bm \xi$ for any measurable function $g(\bm \xi)$.
The local Maxwellian equilibrium $\mathcal{E}[f]$ is implicitly defined by $f$ through the macroscopic fluid density $\rho$, velocity $\bm U \in \mathbb{R}^D$, energy $E$ (or temperature $\theta$) which are the velocity moments of $f$.

It can be easily verified that $\mathcal{E}[f]$ reproduces the local macroscopic quantities \cite{Shar2015}:
\begin{equation} \label{eq:equil_mac}
    \langle \bm m(\bm \xi) \mathcal{E}[f] \rangle = \bm \rho := \rho \left( 1, \bm U, E \right)^T \in \mathbb{R}^{D+2}.
\end{equation}
Moreover, it was shown \cite{Mieu2000} that given any $\bm \rho$ with $\rho>0$ and $\theta>0$, $\mathcal{E}[f]$ is the unique solution that minimizes the following kinetic entropy $H[f]$ :
\begin{equation} \label{eq:entropy}
    H[f] = \langle f \ln f - f \rangle
\end{equation}
subject to the constraint Eq.~(\ref{eq:equil_mac}).

Solving the multidimensional BGK Eq.~(\ref{eq:bgk}) can be computationally costly.
In this work, we propose a class of discrete-velocity-direction models (DVDM) with a minimum entropy.

In the DVDM, the particle transport is limited to $N$ prescribed directions $\bm l_m \in \mathbb{R}^D$ with $|\bm{l}_m |=1$ ($m=1,\dots,N$). Denote by $\bm l_{\mathcal{L}} = \left(\bm{l}_m \right)_{m=1}^N$.
Usually, we take $N>D$ and choose the directions such that the matrix $L^T = (\bm l_1, \dots, \bm l_N) \in \mathbb{R}^{D\times N}$ has rank $D$.
The velocity distribution $f(t,\bm x,\bm \xi)$ is replaced by $N$ 1-D velocity distributions $f_{\mathcal{L}} = \{ f_m(t,\bm x, \xi) \}_{m=1}^N$ with $\xi \in \mathbb{R}$.
The governing equation for each $f_m=f_m(t,\bm x, \xi)$ has the following form:
\begin{equation} \label{eq:dvdm_bgk}
    \partial_t f_m + \xi \bm{l}_m \cdot \nabla_{\bm x} f_m = \frac{1}{\tau} (\mathcal{E}_m - f_m)
\end{equation}
with the local equilibrium $\mathcal{E_L}=\{\mathcal{E}_m(t,\bm x, \xi)\}_{m=1}^N$ yet to be determined by $f_{\mathcal{L}}$.

For further references, we introduce a bracket $\langle (\cdot)_{\mathcal{L}} \rangle_{\mathcal{L}}$ for any $N$-tuple $g_{\mathcal{L}}=\left( g_m(\xi) \right)_{m=1}^N$ as
\begin{equation} \label{eq:operator_dvdm}
    \langle g_{\mathcal{L}} \rangle_{\mathcal{L}} = \sum_{m=1}^N \int_{\mathbb{R}} g_m(\xi) d\xi.
\end{equation}
Define
\[
\rho=\langle f_{\mathcal{L}} \rangle_{\mathcal{L}}, \quad
\rho \bm U=\langle \xi \bm l_{\mathcal{L}} f_{\mathcal{L}} \rangle_{\mathcal{L}}, \quad
\rho E=\langle \frac{\xi^2}{2} f_{\mathcal{L}} \rangle_{\mathcal{L}},
\]
or equivalently
\begin{equation}  \label{eq:macpara_dvdm}
    \bm \rho = \langle \bm m_{\mathcal{L}} f_{\mathcal{L}} \rangle_{\mathcal{L}},
\end{equation}
with $\bm m_{\mathcal{L}}=\{ \bm m_m \}_{m=1}^N$ and $\bm{m}_m = \left(1, \xi \bm{l}_m, \frac{1}{2} \xi^2 \right)^T \in \mathbb{R}^{D+2}$ for $m=1,\dots,N$.

The local equilibrium $\mathcal{E_L}$ is determined so that it minimizes the discrete analogue of the kinetic entropy
\begin{equation} \label{eq:discEntropy_def}
    H[f_{\mathcal{L}}] = \langle f_{\mathcal{L}} \ln f_{\mathcal{L}} - f_{\mathcal{L}} \rangle_{\mathcal{L}}.
\end{equation}
among all possible $N$-tuples $f_{\mathcal{L}} \ge 0$ satisfying $\langle \bm m_{\mathcal{L}} f_{\mathcal{L}} \rangle_{\mathcal{L}} = \bm \rho$ for a given $\bm \rho$.
The same idea was utilized to develop a conservative and entropy-decreasing discrete-velocity model (DVM) of the BGK equation \cite{Mieu2000}.

As to the existence and uniqueness of the minimizer, we have the following result.
\begin{theorem} \label{thm:alpha}
    Suppose $f_{\mathcal{L}}=\{f_m(\xi) \}_{m=1}^N$ and $\bm \rho = \langle \bm m_{\mathcal{L}} f_{\mathcal{L}} \rangle_{\mathcal{L}}$ satisfy $f_m \geq 0$ for all $m$ and $0< |\bm \rho | < \infty$.
    Then the discrete kinetic entropy Eq.~(\ref{eq:discEntropy_def}) has a unique minimizer $\mathcal{E_L}$ satisfying the constraint $\langle \bm m_{\mathcal{L}} \mathcal{E_L} \rangle_{\mathcal{L}} = \bm \rho$.
    Moreover, the minimizer has the exponential form
    \begin{equation} \label{eq:discreteeq}
        \mathcal{E}_m = \exp(\bm \alpha \cdot \bm{m}_m)
    \end{equation}
    with a certain $\bm \alpha\in \mathbb{R}^{D+2}$.
\end{theorem}

The proof of this theorem will be completed in the next section, where we show that there exists $\bm \alpha= (\alpha_0,\hat{\bm{\alpha}}^T, \alpha_{D+1})^T \in \mathbb{R}^{D+2}$ with $\hat{\bm \alpha} \in \mathbb{R}^D$ such that
\begin{equation}\label{eq:alpha_def}
    \langle \bm{m}_{\mathcal{L}} \exp \left(\bm \alpha \cdot \bm{m}_{\mathcal{L}} \right) \rangle_{\mathcal{L}} = \bm \rho.
\end{equation}
This implies $\alpha_{D+1}<0$.
Having such an $\bm \alpha$, we can show that the local equilibrium $\mathcal{E_L}$ with $\mathcal{E}_m = \exp(\bm \alpha \cdot \bm m_m)$ is the unique minimizer.

To do this, we use the fact that the function $(x\ln x - x)$ is strictly convex. Then, for any $f_{\mathcal{L}} = \{f_m\}_{m=1}^N$ satisfying $\langle \bm{m}_{\mathcal{L}} f_{\mathcal{L}} \rangle_{\mathcal{L}} = \bm \rho$,
we have
\[
f_m \ln f_m - f_m \geq \mathcal{E}_m \ln \mathcal{E}_m - \mathcal{E}_m + (\ln \mathcal{E}_m)(f_m-\mathcal{E}_m) = \mathcal{E}_m \ln \mathcal{E}_m - \mathcal{E}_m + (\bm \alpha \cdot \bm{m}_m)(f_m-\mathcal{E}_m)
\]
for all $m$, and thereby,
\begin{align*}
    H[f_{\mathcal{L}}]
    & \geq H[\mathcal{E}_{\mathcal{L}}] + \langle (\bm \alpha \cdot \bm{m}_{\mathcal{L}})(f_{\mathcal{L}}-\mathcal{E_L}) \rangle_{\mathcal{L}} \\
    & = H[\mathcal{E}_{\mathcal{L}}] + \alpha_0 \langle f_{\mathcal{L}}-\mathcal{E_L} \rangle_{\mathcal{L}} + \hat{\bm \alpha} \cdot \langle \xi \bm{l}_{\mathcal{L}} (f_\mathcal{L}-\mathcal{E_L}) \rangle_{\mathcal{L}} + \alpha_{D+1} \langle \frac{\xi^2}{2} (f_\mathcal{L}-\mathcal{E_L}) \rangle_{\mathcal{L}} \\
    & = H[\mathcal{E}_{\mathcal{L}}].
\end{align*}
The equality holds if and only if $f_m = \mathcal{E}_m$ for all $m = 1,...,N$.
Thus, if there is another minimizer of the exponential form, say, $\exp(\bm \alpha' \cdot \bm m_m)$, then $0 = (\bm \alpha' - \bm \alpha) \cdot \bm m_m = (\alpha_0' - \alpha_0) + \xi \bm l_m \cdot (\hat{\bm \alpha}' - \hat{\bm \alpha}) + \frac{\xi^2}{2}(\alpha_{D+1}'-\alpha_{D+1})$,
which means $\bm \alpha'=\bm \alpha$ by choosing different $\xi$ and using the assumption that $L^T=(\bm l_1,\dots,\bm l_N)$ is of full-rank. Hence, the uniqueness of the minimizer is proved.

With the local equilibrium $\mathcal{E_L}$ determined as above, the DVDM Eq.~(\ref{eq:dvdm_bgk}) is well defined as
\begin{equation} \label{eq:dvdm_bgk_full}
    \left \{
    \begin{aligned}
        \partial_t f_m + \xi \bm{l}_m \cdot \nabla_{\bm x} f_m &= \frac{1}{\tau} (\exp(\bm \alpha
        \cdot \bm m_m) - f_m), \quad m=1,\dots,N, \\
        \langle \bm m_{\mathcal{L}} \exp (\bm \alpha \cdot \bm m_{\mathcal{L}}) \rangle_{\mathcal{L}} &= \bm \rho := \langle \bm m_{\mathcal{L}} f_{\mathcal{L}} \rangle_{\mathcal{L}}.
    \end{aligned}
    \right.
\end{equation}
The $N$ differential equations are coupled through $\bm \alpha$ determined by the last $(D+2)$ nonlinear algebraic equations in terms of $f_{\mathcal{L}}=\{f_m\}_{m=1}^N$.
Note that the local equilibrium in Eq.~(\ref{eq:dvdm_bgk_full}) can be rewritten as
\begin{equation} \label{eq:dvdm_gau}
    \mathcal{E}_m = \exp(\bm \alpha \cdot \bm{m}_m)
    = \frac{\rho_m}{\sqrt{2\pi} \sigma} \exp \left(-\frac{(\xi - u_m)^2}{2\sigma^2} \right)
\end{equation}
for $\alpha_{D+1}<0$. The parameter $(\rho_m,  u_m, \sigma^2)$ is related to $\bm \alpha$ as follows:
\begin{equation}  \label{eq:stand_rel}
    \sigma^2 = -\frac{1}{\alpha_{D+1}}, \quad
    u_m= \bm{l}_m \cdot \left( \sigma^2 \hat{\bm \alpha} \right), \quad
    \rho_m= \sqrt{2\pi} \sigma \exp \left( \alpha_0 + \frac{u_m^2}{2\sigma^2} \right).
\end{equation}

\begin{remark} \label{rem:dvdm_dirNcons}
    While $\mathcal{E_L}$ preserves the local fluid quantity $\bm \rho = \langle \bm m_{\mathcal{L}} f_{\mathcal{L}} \rangle_{\mathcal{L}}$, the conservation property generally does not hold in any specific direction.
    That is to say, $\rho_m \ne \int_{\mathbb{R}} f_m d\xi$, $\rho_m u_m \ne \int_{\mathbb{R}} \xi f_m d\xi$, and $\rho_m (u_m^2 + \sigma^2) \ne \int_{\mathbb{R}} \xi^2 f_m d\xi$ for $m=1,\dots,N$.
    Indeed, this is the key mechanism reallocating molecules among the prescribed directions in the DVDM-BGK model Eq.~(\ref{eq:dvdm_bgk_full}).
\end{remark}

The DVDM-BGK model Eq.~(\ref{eq:dvdm_bgk_full}) has also the following properties including an $H$-theorem.

\begin{theorem}
    Suppose the DVDM-BGK Eq.~(\ref{eq:dvdm_bgk_full})
    with positive initial data $f_m(0,\bm x, \xi)>0$ has a solution $f_{\mathcal{L}} = \{ f_m(t,\bm x, \xi) \}_{m=1}^N$.
    Then we have
    \begin{align}
        f_m(t,\bm x,\xi) > 0, \quad \forall \ m,t, &\bm x,\xi, \label{eq:dvdm_pos} \\
        \partial_t \langle \bm m_{\mathcal{L}} f_{\mathcal{L}} \rangle_{\mathcal{L}} + \langle \bm m_{\mathcal{L}} \xi \bm l_{\mathcal{L}} \cdot \nabla_{\bm x}f_{\mathcal{L}} \rangle_{\mathcal{L}} &= 0, \label{eq:dvdm_consv} \\
        \partial_t H[f_{\mathcal{L}}] + \langle \xi \bm l_{\mathcal{L}} \cdot \nabla_{\bm x} (f_{\mathcal{L}} \ln f_{\mathcal{L}} - f_{\mathcal{L}}) \rangle_{\mathcal{L}} &\leq 0. \label{eq:dvdm_Hthm}
    \end{align}
\end{theorem}

\begin{proof}
    For Eq.~(\ref{eq:dvdm_pos}), we define $g_m(t,\bm x,\xi) = f_m(t,\bm x + \xi t \bm l_m,\xi)$ and deduce from Eq.~(\ref{eq:dvdm_bgk_full}) that
    \[
    \begin{aligned}
        \frac{d}{dt}g_m(t,\bm x,\xi) &= \partial_t f_m(t,\bm x + \xi t \bm l_m,\xi) + \xi \bm l_m \cdot \nabla_{\bm x} f_m(t,\bm x + \xi t \bm l_m,\xi) \\
        &= \frac{1}{\tau} \left (\mathcal{E}_m (t,\bm x + \xi t \bm l_m,\xi) - g_m (t,\bm x,\xi) \right ),
    \end{aligned}
    \]
    yielding
    \[
    \frac{d}{dt} \left( e^{\frac{t}{\tau}} g_m(t,\bm x,\xi) \right) = \frac{1}{\tau} e^{\frac{t}{\tau}} \mathcal{E}_m (t,\bm x + \xi t \bm l_m,\xi).
    \]
    Solving $g_m(t,\bm x,\xi)$ from the above equation, we obtain
    \[
    f_m(t,\bm x,\xi) = \frac{1}{\tau} \int_0^t e^{\frac{s-t}{\tau}} \mathcal{E}_m (s,\bm x + \xi (s-t) \bm l_m,\xi) \mathrm{d}s + e^{\frac{-t}{\tau}} f_m(0,\bm x - \xi t \bm l_m,\xi).
    \]
    Thus, Eq.~(\ref{eq:dvdm_pos}) follows immediately from the positivity of $f_m(0,\bm x, \xi)$ and $\mathcal{E}_m = \exp (\bm \alpha \cdot \bm m_m)$.

    To derive Eq.~(\ref{eq:dvdm_consv}), we just exert the operation $\langle \bm m_{\mathcal{L}} (\cdot)_{\mathcal{L}} \rangle_{\mathcal{L}}$ on both sides of the differential equations of Eq.~(\ref{eq:dvdm_bgk_full}) and notice that the RHS vanishes as
    $\langle \bm m_{\mathcal{L}} \left( \exp(\bm \alpha \cdot \bm m_{\mathcal{L}}) -f_{\mathcal{L}} \right) \rangle_{\mathcal{L}}=0$.

    For Eq.~(\ref{eq:dvdm_Hthm}), we multiply the two sides of Eq.~(\ref{eq:dvdm_bgk}) with $\ln f_m$ to obtain
    \[
    \begin{aligned}
        \partial_t (f_m \ln f_m - f_m) + \xi \bm l_m \cdot \nabla_{\bm x} (f_m \ln f_m - f_m) &= \frac{1}{\tau} (\mathcal{E}_m - f_m) (\ln f_m) \\
        & \leq \frac{1}{\tau} (\mathcal{E}_m - f_m) (\ln \mathcal{E}_m),
    \end{aligned}
    \]
    where the inequality $(x-y)(\ln x - \ln y) \ge 0$ has been used. Then, Eq.~(\ref{eq:dvdm_Hthm}) follows by applying $\langle (\cdot) \rangle_{\mathcal{L}}$ to both sides of the last inequality and noticing that $\langle (\mathcal{E}_{\mathcal{L}} - f_{\mathcal{L}}) (\ln \mathcal{E}_{\mathcal{L}}) \rangle_{\mathcal{L}} = 0$.
    This completes the proof.
\end{proof}

Note that Eq.~(\ref{eq:dvdm_consv}) is just the classical conservation laws for the macroscopic fluid quantities:
\[
\partial_t \bm \rho + \nabla_{\bm x} \cdot
\begin{pmatrix}
    \rho \bm U \\
    \rho \bm U \otimes \bm U + \mathbf{P} \\
    \rho \bm U E + \mathbf{P} \cdot \bm U + \bm q
\end{pmatrix}
= \bm 0.
\]
Here $\mathbf{P} = \left \langle (\xi \bm{l}_{\mathcal{L}} - \bm U)^{\otimes 2} f_{\mathcal{L}} \right \rangle_{\mathcal{L}}$ and
$\bm q = \frac{1}{2} \langle (\xi \bm{l}_{\mathcal{L}} - \bm U) |\xi \bm{l}_{\mathcal{L}} - \bm U|^2 f_{\mathcal{L}} \rangle_{\mathcal{L}}$.

\begin{remark} [Planar flows]
  For real-world planar rarefied flows, we have $D=3$, $\bm \xi = (\xi_x,\xi_y,\xi_z)^T$ and $D_x=2$. In DVDM, it is straightforward to select velocity orientations $\bm l_m$ in $\mathbb{R}^3$ to obtain the governing Eq.~(\ref{eq:dvdm_bgk}) \cite{Zhang2008}. On the other hand, before discretizing velocity directions, we can also resort to the technique of reduced distribution functions to eliminate the dependence of $\xi_z$.
  There are two approaches for this purpose: ({\it i}) Define $(g,h)^T = \int (1, \xi_z^2)^T f d\xi_z$ \cite{Chu1965}; ({\it ii}) Define $(g,h)^T = \frac{1}{\sqrt{\pi}} \int \left(e^{-\xi_z^2}, \left( \xi_z^2-\frac{1}{2} \right)e^{-\xi_z^2} \right)^T f d\xi_z$ \cite{Shar2015}.
  In both cases we have $g = g(t,\bm x, \xi_x, \xi_y)$ and $h = h(t,\bm x, \xi_x, \xi_y)$. The governing equations for $g$ and $h$ can be derived from the BGK equation and the orientations $\bm l_m$ are hence chosen on the $(\xi_x,\xi_y)$-plane. The discrete equilibriums can be modeled and solved out of the minimum entropy principle in a similar manner.
\end{remark}

We conclude this section with the following proposition for an alternative way to solve $\bm \alpha$ from the nonlinear algebraic equations Eq.~(\ref{eq:alpha_def}).

\begin{proposition}\label{prop1}
    $\bm \alpha$ solves Eq.~(\ref{eq:alpha_def}) if and only if $\bm \alpha$ minimizes
    \begin{equation} \label{eq:J_def}
        J(\bm \alpha) := \langle \exp \left( \bm \alpha \cdot \bm{m}_{\mathcal{L}} \right) \rangle_{\mathcal{L}} - \bm \rho \cdot \bm \alpha
    \end{equation}
    for $\bm\alpha=(\alpha_0,\hat{\bm \alpha},\alpha_{D+1})$ with $\alpha_{D+1}<0$.
\end{proposition}

\begin{proof}
    If $\bm \alpha$ is a minimizer, we have $\nabla J (\bm \alpha) = 0$, which is just Eq.~(\ref{eq:alpha_def}).
    Conversely, if $\bm \alpha$ solves Eq.~(\ref{eq:alpha_def}), then we have $\nabla J(\bm \alpha)=0$. Thus, it suffices to show that $J=J(\bm \alpha)$ is strictly convex. To do so, we compute the Hessian
    \[
    \nabla_{\bm \alpha \bm \alpha} J = \langle \bm{m}_{\mathcal{L}} \otimes \bm{m}_{\mathcal{L}} \exp \left(\bm \alpha \cdot \bm{m}_{\mathcal{L}} \right) \rangle_{\mathcal{L}}
    \]
    and consider the quadratic form
    \[
    \bm{\eta}^T \nabla_{\bm \alpha \bm \alpha} J \bm \eta = \langle (\bm{m}_{\mathcal{L}} \cdot \bm{\eta})^2 \exp \left(\bm \alpha \cdot \bm{m}_{\mathcal{L}} \right) \rangle_{\mathcal{L}} \geq 0
    \]
    for any $\bm \eta \in \mathbb{R}^{D+2}$.
    The equality holds if and only if $\bm{m}_m \cdot \bm \eta = 0$ for all $m=1,\dots, N$ and thus $\bm \eta = \bm 0$ due to the assumption that $L^T=(\bm l_1,...,\bm l_N)$ is of full-rank. This indicates that the Hessian $\nabla_{\bm \alpha \bm \alpha} J$ is positive definite and therefore $J$ is strictly convex.
    This completes the proof.
\end{proof}

\begin{remark}
    This proposition does not claim the existence of the solution to Eq.~(\ref{eq:alpha_def}). But the strict convexity of $J=J(\bm \alpha)$ implies the uniqueness.
\end{remark}

\section{Existence of $\alpha$} \label{sec:dis_equil}

This section is devoted to completing the proof of Theorem~\ref{thm:alpha}. Namely, we will prove that $J(\bm \alpha)$ defined in Eq.~(\ref{eq:J_def})
attains its minimum in the open set $\mathcal{D}:= \mathbb{R} \times \mathbb{R}^D \times \mathbb{R}_-$. Throughout this section, we assume that the given $f_{\mathcal{L}}$ and $\bm \rho = \langle \bm m_{\mathcal{L}} f_{\mathcal{L}} \rangle_{\mathcal{L}}$ satisfy the constraints in Theorem~\ref{thm:alpha}: $f_m \ge 0$ for all $m$ and $0<| \bm \rho | < \infty$.

First of all, we give a further constraint on $\bm \rho$.

\begin{lemma} \label{prop:rhoimage}
    Let $\bm \rho = (\rho,\rho \bm U, \rho E)^T \in \mathbb{R}^{D+2}$ be computed from $ f_{\mathcal{L}} = \{f_m\}_{m=1}^N$ in terms of $\bm \rho = \langle \bm m_{\mathcal{L}} f_{\mathcal{L}} \rangle_{\mathcal{L}}$. Then there is a constant $s_{min} \ge 1$ such that
    \begin{equation*}
        \rho > 0,\quad
        E > \frac{1}{2}U^2 s_{min}^2.
    \end{equation*}
\end{lemma}

\begin{proof}
    Set $\hat{\rho}_m=\int_{\mathbb{R}} f_m d\xi \ge 0$ and $\hat{\rho}_m \hat{u}_m = \int_{\mathbb{R}} \xi f_m d\xi$ for $m=1,\dots,N$. We deduce from $\bm \rho = \langle \bm m_{\mathcal{L}} f_{\mathcal{L}} \rangle_{\mathcal{L}}$ and the Cauchy-Schwartz inequality that
    \begin{equation}  \label{eq:relaxed_U}
        \rho = \sum_{m=1}^N \hat{\rho}_m >0, \quad
        \rho \bm U = \sum_{m=1}^N \hat{\rho}_m \hat{u}_m \bm l_m,
    \end{equation}
    \[
    2\rho E = \sum_{m=1}^N \int_{\mathbb{R}} \xi^2 f_m d\xi > \sum_{m: \ \hat \rho_m > 0} \frac{\left( \int_{\mathbb{R}} \xi f_m d\xi \right)^2}{\int_{\mathbb{R}}f_m d\xi} = \sum_{m=1}^N \hat{\rho}_m \hat{u}_m^2.
    \]
    The inequality is strict because the equality holds only when $f_m=0$ for all $m$, which contradicts $\rho>0$. Since
    \[
    \begin{aligned}
        \sum_{m=1}^N \hat{\rho}_m \hat{u}_m^2
        & = \frac{1}{\rho} \sum_{m,n} \hat{\rho}_n \hat{\rho}_m \hat{u}_m^2
        = \frac{1}{\rho} \left[ \sum_m (\hat{\rho}_m \hat{u}_m)^2 + \sum_{m < n} \hat{\rho}_m \hat{\rho}_n ( \hat{u}_m^2 + \hat{u}_n^2 ) \right] \\
        & \ge \frac{1}{\rho} \left( \sum_{m=1}^N |\hat{\rho}_m \hat{u}_m| \right)^2,
    \end{aligned}
    \]
    we have
    \begin{equation}\label{eq:3.2}
        2\rho E > \frac{1}{\rho} \left( \sum_{m=1}^N |\hat{\rho}_m \hat{u}_m| \right)^2.
    \end{equation}
    On the other hand, the second equation in Eq.~(\ref{eq:relaxed_U}) can be rewritten as
    \[
    \sum_{m=1}^N \left(\frac{\hat{\rho}_m \hat{u}_m}{\rho U} \right) \bm l_m = \bm e := \frac{\bm U}{U}.
    \]
    Here we have assumed that $U \neq 0$, otherwise the proof is already complete. Denote by $s_{min} = s_{min}(\bm e)$ the minimum of $\sum_{m=1}^N |a_m|$ subject to the constraints $\sum_{m=1}^N a_m \bm l_m = \bm e \in \mathbb{R}^D$. Obviously we have $s_{min}(\bm e) \ge 1$ because
    \[
    \sum_{m=1}^N |a_m| = \sum_{m=1}^N | a_m \bm l_m | \ge \left | \sum_{m=1}^N a_m \bm l_m \right | = | \bm e |=1.
    \]
    Hence we see that $\sum_{m=1}^N |\hat{\rho}_m \hat{u}_m| \ge \rho U s_{min}(\bm e)$ and the proof is complete.
\end{proof}

\begin{remark}
    Lemma~\ref{prop:rhoimage} indicates the capacity of DVDM to realize macroscopic flow states.
    In the Boltzmann equation, the constraint for $\bm \rho$ is $\rho>0$ and $E>\frac{1}{2}U^2$ (because the temperature $\theta = (2E - U^2)/D>0$).
    By contrast, Lemma~\ref{prop:rhoimage} gives a bigger lower bound for $E$ because $s_{min}\ge 1$, exhibiting the price we pay in the DVDM approximation. Generally speaking, more states can be realized with more discrete velocity directions (see Appendix \ref{sec:appendix_realize} for the case $D=2$). On the other hand, no upper bound is required for $E$ or $\theta$ in DVDM, which is in contrast to DVM \cite{Mieu2000}.
\end{remark}

To show the existence of a minimizer $\bm \alpha$ of $J(\bm \alpha)$, we take $\bm \alpha^* \in \mathcal{D}$ and set $B_M := \{\bm \alpha \in \mathcal{D}: J(\bm \alpha) \le J(\bm \alpha^*)\}$. Thus it suffices to prove that $B_M$ is compact, namely, $B_M$ is both closed in $\mathbb{R}^{D+2}$ and bounded.

The closedness is due to the following lemma, which is similar to Proposition 2 (P2) in Ref.~\cite{Mieu2001} but its proof needs more efforts.
\begin{lemma} \label{lem:alpha_to_0}
    If a sequence $\{\bm \alpha^{(p)}\}_{p=1}^{\infty} \subset \mathcal{D}$ satisfies
    \[
    \lim_{p\to \infty} \alpha_{D+1}^{(p)} = 0,
    \]
    then there is a subsequence $\{\bm \alpha^{(p_k)}\}_{k=1}^{\infty}$ such that $J(\bm \alpha^{(p_k)})$ goes to $+\infty$.
\end{lemma}
\begin{proof}
    Note that $\|\bm b\|_{\mathcal{L}} = \sqrt{\sum_{m=1}^N (\bm l_m \cdot \bm b)^2}$ is a norm on $ \mathbb{R}^D$. The reason for this is that $L^T = (\bm l_1,...,\bm l_N) \in \mathbb{R}^{D\times N}$ is full-rank, $L^T L$ is positive definite, and $\sum_{m=1}^N (\bm l_m \cdot \bm b)^2 = \bm b^T L^T L \bm b$. Thus there exists a constant $C>0$ such that $\|\bm b\|_{\mathcal{L}}^2 \ge C | \bm b |^2$ for all $\bm b\in \mathbb{R}^D$. In what follows, we use $C$ as a generic constant.

    A direct computation from Eq.~(\ref{eq:J_def}) yields
    \begin{equation}\label{eq:J_detail}
        J(\bm \alpha) = \sum_{m=1}^N \sqrt{\frac{2\pi}{-\alpha_{D+1}}} \exp \left( \alpha_0 + \frac{(\hat{\bm \alpha} \cdot \bm{l}_m)^2}{2(-\alpha_{D+1})} \right) - \rho (\alpha_0 + \hat{\bm \alpha} \cdot \bm U + \alpha_{D+1} E).
    \end{equation}
    Since $e^x$ is a convex function, we have
    \[
    \sum_{m=1}^N \exp \left( \frac{(\hat{\bm \alpha} \cdot \bm{l}_m)^2}{2(-\alpha_{D+1})} \right)
    \ge
    N \exp \left( \frac{\sum_{m=1}^N(\hat{\bm \alpha} \cdot \bm{l}_m)^2}{2N(-\alpha_{D+1})} \right)
    \ge
    C \exp \left( C \frac{|\hat{\bm \alpha}|^2}{|\alpha_{D+1}|} \right)
    \]
    and therefore
    \begin{equation} \label{eq:estimate_J}
        J(\bm \alpha)
        \geq \frac{C}{\sqrt{|\alpha_{D+1}|}} \exp \left( \alpha_0 + C \frac{|\hat{\bm \alpha}|^2}{|\alpha_{D+1}|} \right) - \rho (\alpha_0 + U|\hat{\bm \alpha} |).
    \end{equation}
    With this inequality, the rest of this proof is to show the right-hand side of Eq.~(\ref{eq:estimate_J}) goes to $+\infty$ as $\alpha_0^{(p)} \to 0^-$, which can be found in Ref.~\cite{Mieu2001} and we omit it here.
\end{proof}

It remains to show that $B_M$ is bounded. Otherwise, there is an unbounded sequence $\{\bm \alpha^{(p)}\}_{p=1}^{\infty} \subset B_M$. Thus we only need to show that $J(\bm \alpha^{(p)})$ has no upper bound. To do this, we formulate the following lemma.

\begin{lemma} \label{lem:uniform_on_ball}
    Set
    \[
    \mathbb{S}^{D+1}_- := \{\bm w =(w_0,\hat{\bm w}, w_{D+1}) \in \mathbb{R}^{D+2}: |\bm w|=1, w_{D+1}<0\}.
    \]
    For any $\bm w \in \mathbb{S}^{D+1}_-$, there is an open neighborhood $ \Omega(\bm w)$ of $\bm w$ in $\mathbb{S}^{D+1}_-$ such that $J(t\bm w') \to +\infty$ uniformly on $\Omega(\bm w)$, namely,
    \[
    \lim_{t\to +\infty} \inf_{\bm w'\in \Omega(\bm w)} J(t\bm w') = +\infty.
    \]
\end{lemma}
\begin{proof}
    From Eq.~(\ref{eq:J_detail}) we have
    \begin{equation} \label{eq:Jtw}
        J(t\bm w) = \sum_{m=1}^N \sqrt{\frac{2\pi}{t |w_{D+1}|}} \exp \left(t w_0 + t \frac{(\hat{\bm w} \cdot \bm{l}_m)^2}{2|w_{D+1}|} \right) - t \bm \rho \cdot \bm w.
    \end{equation}
    If there exists $m$ such that $ (\hat{\bm w} \cdot \bm{l}_m)^2 > 2 w_0 w_{D+1} $, then $\bm w$ has a neighborhood $ \Omega(\bm w)$ in $\mathbb{S}^{D+1}_-$ such that for all $\bm w' \in \Omega(\bm w)$,
    \[
    w_0'+ \frac{(\hat{\bm w}' \cdot \bm{l}_m)^2}{2 |w_{D+1}'|} \ge d >0 \quad \text{and} \quad \bm \rho \cdot \bm w' \le C(\bm w) < + \infty.
    \]
    Then it follows from Eq.~(\ref{eq:Jtw}) that $J(t\bm w') \ge \frac{C}{\sqrt{t}} \exp(t d) - t C(\bm w) \to + \infty $ uniformly for $\bm w' \in \Omega(\bm w)$.

    Otherwise, we have $ (\hat{\bm w} \cdot \bm{l}_m)^2 \le 2 w_0 w_{D+1} $ for all $m$. If $w_0=0$, $\hat{\bm w}$ must be $\bm 0$ and we have $\bm \rho \cdot \bm w = \rho(w_0 + \bm U \cdot \hat{\bm w} + E w_{D+1})<0$.
    Otherwise, $w_0$ must be negative. Due to Lemma~\ref{prop:rhoimage}, we can find $\{a_m\}_{m=1}^N \subset \mathbb{R}$ satisfying $U\sum_{m=1}^N a_m \bm l_m = \bm U$ and $\sqrt{2E} > U \sum_{m=1}^N |a_m|$. Then we deduce that
    \[
    \begin{aligned}
        -w_0 - E w_{D+1}
        & \ge 2 \sqrt{E w_0 w_{D+1}}
        >
        U \sum_{m=1}^N |a_m| \sqrt{2  w_0 w_{D+1}}
        \\
        & \ge
        U\sum_{m=1}^N |a_m| |\hat{\bm w} \cdot \bm{l}_m|
        \ge
        U| \sum_{m=1}^N a_m \hat{\bm w} \cdot \bm{l}_m| = |\hat{\bm w} \cdot \bm U|.
    \end{aligned}
    \]
    Therefore, we also have $\bm \rho \cdot \bm w <0$. Now we can find a neighborhood $ \Omega(\bm w)$ of $\bm w$ in $\mathbb{S}^{D+1}_-$ such that $\bm \rho \cdot \bm w' \le -d < 0$ for all $\bm w' \in \Omega(\bm w)$. Thus we have $ J(t\bm w') > t d \to + \infty $ uniformly for $\bm w' \in \Omega(\bm w)$. This completes the proof.
\end{proof}

\begin{remark}
    Since $\mathbb{S}_-^{D+1}$ is not compact, this lemma cannot directly imply that $J(t \bm w)$ goes to infinity uniformly for all directions $\bm w$. In contrast, in DVM \cite{Mieu2000}, the compactness of $\mathbb{S}^{D+1}$ helps to find a finite open cover and prove the coercivity.
\end{remark}

Having Lemma~\ref{lem:uniform_on_ball}, we first suppose $\{\bm \alpha^{(p)}\}$ is contained in a cone $C_r$ defined by
\[
C_r := \left\{ \bm \alpha \in \mathcal{D} : \sqrt{\alpha_0^2 + |\hat{\bm \alpha}|^2} \le r |\alpha_{D+1}| \right\}
\]
for some $r>0$. Then $\mathbb{S}_r^{D+1}:= \mathbb{S}_-^{D+1} \cap C_r$ is a compact set. Lemma~\ref{lem:uniform_on_ball} implies that for any $\bm w \in \mathbb{S}_r^{D+1}$, there exist a neighborhood $ \Omega(\bm w)$ in $\mathbb{S}^{D+1}_r$ and $T_{\bm w}>0$, such that
\[
J(t \bm w') > J(\bm \alpha^*) + 1 \ \text{for all } t>T_{\bm w},\ \bm w'\in \Omega(\bm w).
\]
Since $\cup_{\bm w \in \mathbb{S}^{D+1}_r}\Omega(\bm w)$ is an open cover of $\mathbb{S}^{D+1}_r$, we can find a finite set $\{\bm w^{(q)}\}_{q=1}^Q \subset \mathbb{S}^{D+1}_r$ such that $\mathbb{S}^{D+1}_r \subset \cup_{q=1}^Q \Omega(\bm w^{(q)})$. We take $T = \max_{1\le q\le Q} T_{\bm w^{q}}$ and get $ J(t \bm w) > J(\bm \alpha^*) + 1$ for all $t>T$ and $\bm w\in \mathbb{S}^{D+1}_r$. Therefore, $\{\bm \alpha^{(p)}\}$ is bounded in $C_r \cap \{ |\bm \alpha| \le T \}$, which is a contradiction.

The last argument indicates that $\sqrt{|\alpha_0^{(p)}|^2 + |\hat{\bm \alpha}^{(p)}|^2} / |\alpha_{D+1}^{(p)}| \to \infty$, which implies that either $|\alpha_0^{(p)}| / |\alpha_{D+1}^{(p)}|$ or $|\hat{\bm \alpha}^{(p)}| / |\alpha_{D+1}^{(p)}|$ is unbounded. Note that we can assume $|\alpha_{D+1}| \ge \delta >0$ due to Lemma~\ref{lem:alpha_to_0}.

To show that $J(\bm \alpha^{(p)})$ is unbounded, we recall Eq.~(\ref{eq:estimate_J}):
\[
J(\bm \alpha) \geq \frac{C}{\sqrt{|\alpha_{D+1}|}} \exp \left( \alpha_0 + C \frac{|\hat{\bm \alpha}|^2}{|\alpha_{D+1}|} \right) - \rho (\alpha_0 +  U |\hat{\bm \alpha}|).
\]
Then we follow Ref.~\cite{Mieu2001} and divide the argument into five cases, where we use $M^{(p)}$ to denote terms that go to $+\infty$ as $p\to \infty$.
\begin{enumerate}
    \item $\{\alpha_0^{(p)}\}$ is bounded. Then $|\alpha_0^{(p)}| / |\alpha_{D+1}^{(p)}|$ is bounded and we have
    \[
    \frac{|\hat{\bm \alpha}^{(p)}|}{|\alpha_{D+1}^{(p)}|} \to \infty \quad \text{and} \quad |\hat{\bm \alpha}^{(p)}| \ge \frac{\delta|\hat{\bm \alpha}^{(p)}|}{|\alpha_{D+1}^{(p)}|} \to \infty.
    \]
    Thus it follows that
    \[
    \begin{aligned}
        J(\bm \alpha^{(p)})
        & \ge \frac{C}{\sqrt{|\alpha_{D+1}^{(p)}|}} \exp \left( C \frac{|\hat{\bm \alpha}^{(p)}|^2}{|\alpha_{D+1}^{(p)}|} \right) - C |\hat{\bm \alpha}^{(p)}| - C \\
        & \ge C \frac{|\hat{\bm \alpha}^{(p)}|^4}{|\alpha_{D+1}^{(p)}|^{\frac{5}{2}}} -C |\hat{\bm \alpha}^{(p)}| - C \to + \infty.
    \end{aligned}
    \]
    \item $\alpha_0^{(p)} \to + \infty$ and $|\alpha_0^{(p)}| / |\alpha_{D+1}^{(p)}| \to + \infty$. In this case, we have
    \[
    \begin{aligned}
        J(\bm \alpha^{(p)})
        & \ge \frac{C}{\sqrt{|\alpha_{D+1}^{(p)}|}} \left( 1 + \frac{|\hat{\bm \alpha}^{(p)}|^2}{|\alpha_{D+1}^{(p)}|} \right)
        \exp( \alpha_0^{(p)} )- C \alpha_0^{(p)} - C |\hat{\bm \alpha}^{(p)}| \\
        & \ge M^{(p)} \alpha_0^{(p)} + M^{(p)} |\hat{\bm \alpha}^{(p)}|^2 - C \alpha_0^{(p)} - C |\hat{\bm \alpha}^{(p)}| \to + \infty.
    \end{aligned}
    \]
    \item $\alpha_0^{(p)} \to + \infty$ and $|\alpha_0^{(p)}| / |\alpha_{D+1}^{(p)}|$ is bounded. As in the first case, we have
    \[
    \frac{|\hat{\bm \alpha}^{(p)}|}{|\alpha_{D+1}^{(p)}|} \to \infty \quad \text{and} \quad |\hat{\bm \alpha}^{(p)}| \ge \frac{\delta|\hat{\bm \alpha}^{(p)}|}{|\alpha_{D+1}^{(p)}|} \to \infty.
    \]
    Thus it follows that
    \[
    \begin{aligned}
        J(\bm \alpha^{(p)})
        & \ge \frac{C}{\sqrt{|\alpha_{D+1}^{(p)}|}} (1 + \alpha_0^{(p)}) \exp \left(C \frac{|\hat{\bm \alpha}^{(p)}|^2}{|\alpha_{D+1}^{(p)}|} \right)
        - C \alpha_0^{(p)} - C |\hat{\bm \alpha}^{(p)}| \\
        & \ge M^{(p)} |\hat{\bm \alpha}^{(p)}| + M^{(p)} \alpha_0^{(p)} - C \alpha_0^{(p)} - C |\hat{\bm \alpha}^{(p)}| \to + \infty.
    \end{aligned}
    \]
    \item $\alpha_0^{(p)} \to - \infty$ and $\alpha_0^{(p)} + U |\hat{\bm \alpha}^{(p)}| \to - \infty$. It is easy to see that $J(\bm \alpha^{(p)}) \to + \infty$ in this case.
    \item $\alpha_0^{(p)} \to - \infty$ and $\alpha_0^{(p)} + U |\hat{\bm \alpha}^{(p)}| \ge -M_L$. Then we have $|\hat{\bm \alpha}^{(p)}| \to + \infty$. Furthermore, we claim $|\hat{\bm \alpha}^{(p)}|/|\alpha_{D+1}^{(p)}| \to \infty$. Otherwise, $|\alpha_0^{(p)}|/|\alpha_{D+1}^{(p)}|$ is unbounded. We can assume $M_L > 0$ and get
    \[
    \frac{\alpha_0^{(p)}}{|\alpha_{D+1}^{(p)}|} + \frac{\hat{\bm \alpha}^{(p)} \cdot \bm U}{|\alpha_{D+1}^{(p)}|} \ge \frac{-M_L}{|\alpha_{D+1}^{(p)}|} \ge \frac{-M_L}{\delta} > - \infty.
    \]
    Thus $|\alpha_0^{(p)}|/|\alpha_{D+1}^{(p)}| \to \infty$ yields $|\hat{\bm \alpha}^{(p)}|/|\alpha_{D+1}^{(p)}| \to \infty$. Now we have
    \[
    \begin{aligned}
        J(\bm \alpha^{(p)})
        & \ge \frac{C}{\sqrt{|\alpha_{D+1}^{(p)}|}} \exp \left( C \frac{|\hat{\bm \alpha}^{(p)}|^2}{|\alpha_{D+1}^{(p)}|} - C |\hat{\bm \alpha}^{(p)}| - M_L \right) - C |\hat{\bm \alpha}^{(p)}| - C \\
        & \ge M^{(p)} \exp \left( \frac{C}{2} \frac{|\hat{\bm \alpha}^{(p)}|^2}{|\alpha_{D+1}^{(p)}|} - C |\hat{\bm \alpha}^{(p)}| \right) - C |\hat{\bm \alpha}^{(p)}| - C \to + \infty.
    \end{aligned}
    \]
\end{enumerate}
Consequently, we have shown that $J(\bm \alpha^{(p)})$ has no upper bound. Hence $B_M$ is bounded and the proof is completed.

We close this section with the following remark.
\begin{remark}
    Recall Eq.~(\ref{eq:stand_rel}). At the discrete equilibrium $\mathcal{E_L}$, we have $\sigma^2 < D \theta$. To see this, we deduce from $\rho \bm U = \sum_{m=1}^N \rho_m u_m \bm l_m$ that
    \[
    \rho U^2 = \frac{1}{\rho} \left( \rho \bm U \cdot \rho \bm U \right) = \frac{1}{\rho} \sum_{m,n}\rho_m \rho_n \left( u_m \bm l_m \cdot u_n \bm l_n \right).
    \]
    Then we have
    \[
    \begin{aligned}
        \sum_m \rho_m u_m^2 - \rho U^2 &= \frac{1}{2\rho} \sum_{m,n} \rho_m \rho_n \left( u_m^2 + u_n^2 - 2 u_m \bm l_m \cdot u_n \bm l_n \right) \\
        &= \frac{1}{2\rho} \sum_{m,n} \rho_m \rho_n \| u_m \bm l_m - u_n \bm l_n \|^2 >0.
    \end{aligned}
    \]
    Thus, the conclusion follows from the relation $\rho U^2 + D \rho \theta = \sum_m \rho_m u_m^2 + \rho \sigma^2$.
\end{remark}

\section{Spatial-time models}
\label{Sec:solve_dvdm}

Besides the spatial-time variables $\bm x$ and $t$, $\xi \in \mathbb{R}$ is another continuous variable in the DVDM model Eq.~(\ref{eq:dvdm_bgk_full}). In this section, we derive models only with $\bm x$ and $t$ as continuous variables by treating $\xi$ in two ways.

\subsection{Discretizing $\xi$}
\label{subsec:dvddvm}
In this first way, the velocity variable $\xi$ in Eq.~(\ref{eq:dvdm_bgk_full}) is replaced with a set of fixed nodes $\xi_{mk} = k \Delta \xi + \xi_0$ for $k=1,\dots,N_m'$ and $\Delta \xi, \xi_0 \in \mathbb{R}$. Namely, each $f_m(t,\bm x,\xi)$ is represented by an $N_m'$-vector $\left( f_{mk}(t,\bm x) \right)_{k=1}^{N_m'}$.
The resultant governing equation for each $f_{mk}=f_{mk}(t,\bm x)$ reads as
\begin{equation} \label{eq:dvddvm}
    \partial_t f_{mk} + \xi_{mk} \bm l_m \cdot \nabla_{\bm x} f_{mk} = \frac{1}{\tau} \left( \mathcal{E}_{mk} - f_{mk} \right)
\end{equation}
for $1\le k\le N_m'$ and $1\le m\le N$. Note that $\Delta \xi$ may depend on the direction $m$ and the index $k$. For the following reason, we denote this kind of models as DVD-DVM.

The DVD-DVM in Eq.~(\ref{eq:dvddvm}) is similar to DVM \cite{Mieu2000} but allows a new way in selecting the discrete velocities. Indeed, common DVM practices use discrete-velocity nodes in a uniform cubic lattice of $\mathbb{R}^D$ \cite{gat1975,Mieu2000}. By contrast, the proposed DVD-DVM creates discrete velocities radially distributed in $\mathbb{R}^D$ and centered at the origin of the velocity space.

To determine the equilibrium $\mathcal{E}_{mk}$ in Eq.~(\ref{eq:dvddvm}), we define the macroscopic fluid quantity $\bm \rho = (\rho, \rho \bm U, \rho E)^T$ as
\begin{equation} \label{eq:dvdm_dvd_macdef}
    \rho = \sum_{m,k}f_{mk} \Delta \xi, \quad
    \rho \bm U = \sum_{m,k} \xi_{mk} f_{mk} \bm l_m \Delta \xi, \quad
    \rho E = \sum_{m,k} \frac{1}{2} \xi_{mk}^2 f_{mk} \Delta \xi.
\end{equation}
For a given $\bm \rho$, we follow the minimum entropy principle and determine the discrete equilibrium $\left\{ \mathcal{E}_{mk} \ge 0 \right\}_{1\le m \le N, 1\le k\le N_m'}$ so that it minimizes the discrete entropy
\begin{equation} \label{eq:dvddvm_discH}
    H[\{g_{mk}\}] = \sum_{m,k} (g_{mk} \ln g_{mk} - g_{mk}) \Delta \xi
\end{equation}
among all possible $\{g_{mk} \ge 0\}_{1\le m \le N, 1\le k\le N_m'}$ satisfying Eq.~(\ref{eq:dvdm_dvd_macdef}). For the existence of such a minimizer, we have the following analogue of Theorem~\ref{thm:alpha}, where $\bm m_{mk} = \left( 1, \xi_{mk} \bm l_m, \frac{1}{2}\xi_{mk}^2 \right)^T \in \mathbb{R}^{D+2}$.
\begin{theorem} \label{thm:dvddvm_equil}
    If $\bm \rho \in \mathbb{R}^{D+2}$ is strictly realizable, i.e., there exists a set $\{g_{mk} > 0\}_{1\le m \le N, 1\le k\le N_m'}$ satisfying Eq.~(\ref{eq:dvdm_dvd_macdef}), then the minimization problem above has a unique solution $\{ \mathcal{E}_{mk} \}$ which can be expressed as
    \begin{equation} \label{eq:dvddvm_equil}
        \mathcal{E}_{mk} = \exp (\bm \alpha' \cdot \bm m_{mk})
    \end{equation}
    with a vector $\bm \alpha' \in \mathbb{R}^{D+2}$. Moreover, $\bm \alpha'$ minimizes the function
    \begin{equation} \label{eq:dvddvm_J}
        J(\bm \alpha) := \sum_{m,k} \exp(\bm \alpha \cdot \bm m_{mk}) \Delta \xi - \bm \rho \cdot \bm \alpha.
    \end{equation}
\end{theorem}

This theorem can be proved in a similar fashion as the conventional DVM with the minimum entropy principle \cite{Mieu2000}. It provides a simple way to compute $\{ \mathcal{E}_{mk} \} \in \mathbb{R}^{\mathcal{N}'}$ with $\mathcal{N}'=\sum_{m=1}^N N_m'$ by solving the $(D+2)$-dimensional minimization problem Eq.~(\ref{eq:dvddvm_J}) without constraints.

Besides, our DVDM in Eq.~(\ref{eq:dvdm_bgk_full}) offers another possibility to determine $\{ \mathcal{E}_{mk} \}$ by conforming `more loosely' to the minimum entropy principle. Recall Eq.~(\ref{eq:dvdm_gau}) that the DVDM equilibrium in each direction $\bm l_m$ is a Gaussian distribution parameterized with $(\rho_m, u_m, \sigma^2)$.
We may choose $\{ \mathcal{E}_{mk} \}$ so that for each $m$, $\{\mathcal{E}_{mk}\}_{1\le k \le N_m'}$ minimizes the discrete entropy
\begin{equation}	\label{eq:dvddvm_discH_dir}
    H_m[\{ g_{mk} \}] = \sum_{k=1}^{N_m'} (g_{mk} \ln g_{mk} - g_{mk}) \Delta \xi
\end{equation}
among all possible $\{ g_{mk} \ge 0\}_{1\le k \le N_m'}$ satisfying
$$
\rho_m = \sum_{k=1}^{N_m'} g_{mk} \Delta \xi, \quad
\rho_m u_m = \sum_{k=1}^{N_m'} \xi_{mk} g_{mk} \Delta \xi, \quad
\rho_m (u_m^2 + \sigma^2) = \sum_{k=1}^{N_m'} \xi_{mk}^2 g_{mk} \Delta \xi.
$$

For these $N$ minization problems, the analogue of Theorem~\ref{thm:dvddvm_equil} holds. Namely, each of them has a unique solution $\{ \mathcal{E}_{mk} \}_{1\le k\le N_m'}$ which can be expressed as
\begin{equation} \label{eq:dvddvm_eq_dir}
    \mathcal{E}_{mk} = \exp \left( \bm \alpha'_m \cdot \bm m'_{mk} \right)
\end{equation}
with a vector $\bm \alpha_m' \in \mathbb{R}^3$ and $\bm m'_{mk} = \left( 1,\xi_{mk}, \frac{1}{2}\xi_{mk}^2 \right)^T$. Moreover, $\bm \alpha_m'$ minimizes the function
$$
J_m(\bm \alpha) := \sum_{k=1}^{N_m'} \exp \left( \bm \alpha \cdot \bm m'_{mk} \right) \Delta \xi - \bm \rho_m' \cdot \bm \alpha
$$
with $\bm \rho_m' = \rho_m \left( 1, u_m, \frac{u_m^2 + \sigma^2}{2} \right)^T \in \mathbb{R}^3$. We note that this approach
does not necessarily find a minimum of the `total' discrete entropy Eq.~(\ref{eq:dvddvm_discH}).

In summary, we have derived two kinds of DVD-DVM models Eq.~(\ref{eq:dvddvm}):
\begin{enumerate}
    \item (DVD-DVM-I). $\mathcal{E}_{mk}$ is given with Eq.~(\ref{eq:dvddvm_equil}) by minimizing the discrete entropy Eq.~(\ref{eq:dvddvm_discH}).
    \item (DVD-DVM-II). $\mathcal{E}_{mk}$ is given with Eq.~(\ref{eq:dvddvm_eq_dir}) by minimizing the discrete entropy Eq.~(\ref{eq:dvddvm_discH_dir}) in each direction. This approach requires $(\rho_m,u_m,\sigma^2)_{m=1}^N$ defined in Eqs.~(\ref{eq:dvdm_gau} \& \ref{eq:stand_rel}) to be computed by DVDM as a first step.
\end{enumerate}

\subsection{Gaussian-EQMOM}
\label{subsec:dvd_eqmom}
Besides the DVD-DVM, the method of moment can be applied to the DVDM Eq.~(\ref{eq:dvdm_bgk_full}). For this purpose, we define the $k$th velocity moment of $f_m(t,\bm x,\xi)$ as
\begin{equation} \label{eq:def_mom}
    M_{m,k}=M_{m,k}(t,\bm x) = \int_{\mathbb{R}} \xi^k f_m(t,\bm x,\xi) d\xi
\end{equation}
for $k\in \mathbb{N}$.
The evolution of $M_{m,k}$ can be derived from Eq.~(\ref{eq:dvdm_bgk_full}) as
\begin{equation} \label{eq:dvdm_mom}
    \partial_t M_{m,k} + \bm{l}_m \cdot \nabla_{\bm x} M_{m,k+1} = \frac{1}{\tau} \left(\rho_m \Delta_k(u_m,\sigma^2) - M_{m,k} \right),
\end{equation}
where $\Delta_k(u,\sigma^2)$ denotes the $k$th moment of the normalized Gaussian function centered at $u$ with a variance $\sigma^2$.
Meanwhile, the macroscopic quantity $\bm \rho$ in Eq.~(\ref{eq:macpara_dvdm}) is computed with
\begin{equation} \label{eq:dvdm_mom_macdef}
    \rho = \sum_{m=1}^N M_{m,0}, \quad
    \rho \bm{U} = \sum_{m=1}^N \bm{l}_m M_{m,1}, \quad
    \rho E = \sum_{m=1}^N \frac{M_{m,2}}{2}.
\end{equation}

There are infinitely many equations in Eq.~(\ref{eq:dvdm_mom}). For each $m$, the first $N'$ equations for moments $M_{m,0}, \dots, M_{m,N'-1}$ are not closed because the $M_{m,N'-1}$-equation contains $M_{m,N'}$ in the convection term.
Hence a closure method is needed.

The Gaussian-EQMOM assumes that the 1-D distribution $f_m(\xi)$ is a sum of $N_m'$ Gaussian functions \cite{MF2013}:
\begin{equation} \label{eq:eqmom_assump}
    f_m(\xi) = \sum_{\alpha=1}^{N_m'} \frac{w_{m,\alpha}}{\sqrt{2\pi \vartheta_m}} \exp \left( -\frac{(\xi-v_{m,\alpha})^2}{2 \vartheta_m} \right).
\end{equation}
Note that the variance $\vartheta_m>0$ is shared among $N_m'$ kernels (and is thus not dependent on $\alpha$).
The weights $w_{m,\alpha}$, nodes $v_{m,\alpha}$ ($\alpha=1,\dots,N_m'$) and variance $\vartheta_m$ are unknown parameters to be solved from the ($2N_m'+1$) nonlinear equations
\begin{equation} \label{eq:eqmom_inv}
    M_{m,k} = \sum_{\alpha=1}^{N_m'} w_{m,\alpha} \Delta_k (v_{m,\alpha},\vartheta_m) \quad\text{for } k=0,\dots,2N_m'.
\end{equation}
An algorithm to solve these equations can be found in the literature \cite{Chalons2017,MF2013}.
Thus, any unknown moments and/or integrals associated with the distribution can be computed by Eq.~(\ref{eq:eqmom_assump}).
For example, the convection term $M_{m,2N_m'+1}$ is reconstructed as
\begin{equation} \label{eq:eqmom_reconstruct}
    \bar{M}_{m,2N_m'+1} = \sum_{\alpha=1}^{N_m'} w_{m,\alpha} \Delta_{2N_m'+1} (v_{m,\alpha},\vartheta_m).
\end{equation}
Consequently, we get
$$
\partial_t M_m + A_m(M_m) \bm l_m \cdot \nabla_{\bm x} M_m = \frac{1}{\tau} (M_{\mathcal{E} m} - M_m),
$$
where
\[
\begin{aligned}
    M_m&=(M_{m,0}, \dots, M_{m,2N_m'})^T \in \mathbb{R}^{2N'_m+1}, \\
    M_{\mathcal{E} m}&=\rho_m \left( \Delta_0(u_m, \sigma^2), \dots, \Delta_{2N_m'}(u_m,\sigma^2) \right)^T \in \mathbb{R}^{2N'_m+1},
\end{aligned}
\]
and
\begin{equation} \label{eq:A_M}
    A_m (M_m) =
    \begin{bmatrix}
        0 & 1 & & & \\
        & 0 & 1 & & \\
        & & \ddots & \ddots & \\
        & & & 0 & 1 \\
        a_{m,0} & a_{m,1} & \cdots & a_{m,2N'_m-1} & a_{m,2N'_m}
    \end{bmatrix}
    \in \mathbb{R}^{(2N_m'+1) \times (2N_m'+1)}
\end{equation}
with $a_{m,k}=\diffp{\bar{M}_{m,2N_m'+1}}{M_{m,k}}$ for $k=0,\dots,2N_m'$.

Putting the last equations together, we derive the following spatial-time model
\begin{equation} \label{eq:dvdm_eqmom}
    \partial_t \mathcal{M} + \sum_{i=1}^{D_x} A^{(i)} \partial_{x_i} \mathcal{M} = \frac{1}{\tau} (M_{\mathcal{E}} - \mathcal{M}).
\end{equation}
Here
\[
\mathcal{M} = \left(M_1^T, \dots, M_N^T \right)^T \in \mathbb{R}^{\mathcal{N}'}, \quad
M_{\mathcal{E}} = \left( M_{\mathcal{E} 1}^T, \dots, M_{\mathcal{E} N}^T \right)^T \in \mathbb{R}^{\mathcal{N}'}
\]
with $\mathcal{N}' = \sum_{m=1}^N (2N_m'+1)$
and
$$
A^{(i)} = \text{diag} \left \{ (\bm l_1 \cdot \bm e_i) A_1(M_1), \dots, (\bm l_N \cdot \bm e_i) A_N(M_N) \right \} \in \mathbb{R}^{\mathcal{N}' \times \mathcal{N}'},
$$
where $\bm e_i$ is the $i$-th column of the identity matrix of order $D_x$.

\begin{theorem}
    The moment system (\ref{eq:dvdm_eqmom}) is hyperbolic.
\end{theorem}

\begin{proof}
    According to Ref.~\cite{HLY2020}, each matrix $A_m(M_m)$ defined in Eq.(\ref{eq:A_M}) is strictly hyperbolic. Namely, it has $(2N_m'+1)$ distinct real eigenvalues.
    Therefore, for any $\bm y = (y_1,...,y_{D_x})^T \in \mathbb{R}^{D_x}$ the block-diagonal matrix $\sum_i y_i A^{(i)}$ is real diagonalizable.
\end{proof}

\begin{remark}
    This approach, denoted as DVD-EQMOM, leads to a convenient multidimensional version of quadrature-based method of moments, which seems better understood than those in Refs.~\cite{Chalons2017,MF2013}.
\end{remark}

\section{Numerical schemes}
\label{sec:numer}
In this section, we show that the DVDM provides new numerical solvers to simulate rarefied flows. The solvers consist of two parts: spatial-time discretization and computation of equilibrium. As the first step of this project, we present the solvers only for spatially one-dimensional models, while the velocity is still in $\mathbb{R}^D$.

\subsection{Spatial-time discretization}
\label{subsec:sptm}
In this part, we discretize the spatial-time models constructed in the previous section. To do this, we denote by $g_{i}^n$ an approximation of function $g(t,x)$ over the grid-block $]x_{i-\frac{1}{2}},x_{i+\frac{1}{2}}[ \times [t_n, t_{n+1})$
with $x_i = i \Delta x$, $t_n = n \Delta t$, $i = 0, \pm 1 , \pm 2,\cdots$, and $n = 0, 1, 2,\cdots$.

\subsubsection{DVD-DVM}
For this kind of models, the governing Eq.~(\ref{eq:dvddvm}) is approximated by the following implicit-explicit scheme
\begin{equation} \label{eq:sch_dvddvm}
    \begin{aligned}
        f_{mk,i}^{n+1} = f_{mk,i}^n - \frac{\Delta t}{\Delta x} \left( \mathcal{F}^n_{mk,i+\frac{1}{2}} - \mathcal{F}^n_{mk,i-\frac{1}{2}} \right)
        + \frac{\Delta t}{\tau} \left( \mathcal{E}^n_{mk,i} - f^{n+1}_{mk,i} \right).
    \end{aligned}
\end{equation}
The numerical flux is taken as
\begin{equation} \label{eq:fluxdvddvm}
    \begin{aligned}
        \mathcal{F}^n_{mk,i+\frac{1}{2}} = \frac{1}{2} \Big(
        & \xi_{mk} \bm l_m \cdot \bm e_1 \left( f^n_{mk,i+1}+f^n_{mk,i} \right) \\
        &- |\xi_{mk} \bm l_m \cdot \bm e_1| \left( f^n_{mk,i+1} - f^n_{mk,i} - \Phi^n_{mk,i+\frac{1}{2}} \right) \Big).
    \end{aligned}
\end{equation}
The flux limiter $\Phi^n_{mk,i+\frac{1}{2}}$ allows to obtain a second-order scheme, and $\Phi^n_{mk,i+\frac{1}{2}}=0$ reduces to a first-order scheme. The equilibrium state $\mathcal{E}^n_{mk,i}$ is solved from the local fluid quantity $\bm \rho^n_{i}$ by either Eq.~(\ref{eq:dvddvm_equil}) (DVD-DVM-I) or Eq.~(\ref{eq:dvddvm_eq_dir}) (DVD-DVM-II).

Note that the source in Eq.~(\ref{eq:sch_dvddvm}) is semi-implicit because $f^{n+1}_{mk,i}$ is used to make the scheme stable, but $f_{mk,i}^{n+1}$ can be explicitly obtained.

\subsubsection{DVD-EQMOM}
For the momemt system, the governing Eq.~(\ref{eq:dvdm_mom}) is approximated by the implicit-explicit scheme:
\begin{equation} \label{eq:sch_dvdeqmom}
    \begin{aligned}
        M_{m,k,i}^{n+1} = & M_{m,k,i}^n - \frac{\Delta t}{\Delta x} \bm l_m \cdot \bm e_1 \left( \mathcal{F}^n_{m,k+1,i+\frac{1}{2}} - \mathcal{F}^n_{m,k+1,i-\frac{1}{2}} \right) \\
        &
        + \frac{\Delta t}{\tau} \left( M^n_{\mathcal{E} m,k,i} - M_{m,k,i}^{n+1} \right).
    \end{aligned}
\end{equation}
The flux is modeled by a `kinetic-based' definition \cite{Chalons2017,MF2013}. For example,
\begin{equation}
    \begin{aligned}
        \mathcal{F}^n_{m,k+1,i+\frac{1}{2}} = \frac{1}{2} \int_{\mathbb{R}} \xi^k \Big[
        & \xi \left(f^n_{m,i+1}+f^n_{m,i} \right) \\
        & - \text{sgn}(\bm l_m \cdot \bm e_1) |\xi| \left(f^n_{m,i+1}-f^n_{m,i}-\Phi^n_{m,i+\frac{1}{2}} \right) \Big] d\xi.
    \end{aligned}
\end{equation}
The limiter $\Phi^n_{m,i+\frac{1}{2}}$ allows to obtain a second-order scheme, and $\Phi^n_{m,i+\frac{1}{2}}=0$ reduces to a first-order scheme. Then we have
\begin{equation} \label{eq:flux_mom_2}
    \mathcal{F}^n_{m,k+1,i+\frac{1}{2}} =\left \{
    \begin{aligned}
        \int_0^{\infty} \xi^{k+1} f_{m,i}^n d\xi + \int_{-\infty}^0 \xi^{k+1} f_{m,i+1}^n d\xi, \quad &\text{if } \bm l_m \cdot \bm e_1 >0, \\
        \int_0^{\infty} \xi^{k+1} f_{m,i+1}^n d\xi + \int_{-\infty}^0 \xi^{k+1} f_{m,i}^n d\xi, \quad &\text{if } \bm l_m \cdot \bm e_1 <0.
    \end{aligned}
    \right.
\end{equation}

To evaluate the integrals above, we define
\[
\langle \xi^{k+1} \rangle^{\pm} (v,\vartheta) = \int_{\xi \gtrless 0} \xi^{k+1} \frac{1}{\sqrt{2\pi \vartheta}} \exp \left( -\frac{(\xi-v)^2}{2\vartheta} \right) d\xi,
\]
which can be computed analytically with recursive relations in terms of $k$ (Eq.~(B.1) in Ref.~\cite{Chalons2017}).
Thus, due to the ansatz Eq.~(\ref{eq:eqmom_assump}), the integrals in Eq.~(\ref{eq:flux_mom_2}) become
\[
\int_{\xi \gtrless 0} \xi^{k+1} f_{m,i} d\xi = \sum_{\alpha=1}^{N_m'} w_{m,\alpha,i} \langle \xi^{k+1} \rangle^{\pm} \left( v_{m,\alpha,i}, \vartheta_{m,i} \right).
\]
The moment inversion algorithm in Refs.~\cite{Chalons2017,MF2013} is used to solve $w_{m,\alpha,i}$, $v_{m,\alpha,i}$, and $\vartheta_{m,i}$ ($\alpha=1,\dots,N_m'$) from the moments $M_{m,0,i},\dots,M_{m,2N_m',i}$.

In the scheme Eq.~(\ref{eq:sch_dvdeqmom}), the moments correspond to the equilibrium state
\[
M^n_{\mathcal{E}m,k,i} = \rho^n_{m,i} \Delta_k \left( u^n_{m,i}, (\sigma^2)^n_{i} \right).
\]
The analytical expression of $\Delta_k(u,\sigma^2)$ is given in Ref.~\cite{HLY2020}. The equilibrium state parameters $\rho^n_{m,i}, \ u^n_{m,i}, \ (\sigma^2)^n_{i}$ are solved by the local fluid quantity $\bm \rho^n_{i}$, as detailed in Section \ref{sec:diseq_comp}.

\subsection{Computation of equilibrium} \label{sec:diseq_comp}
This subsection presents details of computing the discrete equilibrium $\mathcal{E_L}$ in Eq.~(\ref{eq:dvdm_gau}) for a given macroscopic quantity $\bm \rho$.
It can be done by minimizing $J(\bm \alpha)$ in Eq.~(\ref{eq:J_def}), which is strictly convex.
In this work we use the gradient descent method \cite{Curry1944} to solve $\mathcal{E_L}$. Future work is needed to develop stable and efficient higher-order optimization algorithms for this purpose.

The algorithm to minimize $J(\bm \alpha)$ is stated below.
\begin{enumerate}
    \item Let the initial value $\bm \alpha^{(0)} \in \mathbb{R}^{D+1} \times \mathbb{R}^-$ be given. It is found that the simple choice of $\bm \alpha_{eq}$ of the continuous equilibrium in Eq.~(\ref{eq:bgk}), even with $U=0$, works reasonably well in our test case.
    \item \textbf{Repeat}.
    \begin{enumerate}[(i)]
        \item	Update $\bm \alpha$ by
        $$\bm \alpha^{(n+1)} = \bm \alpha^{(n)} - l_r \nabla_{\bm \alpha} J(\bm \alpha^{(n)}).$$
        The gradient $\nabla_{\bm \alpha} J$ has the explicit form as:
        \[
        \nabla_{\bm \alpha} J = \left( \sum_{m=1}^N \rho_m, \ \sum_{m=1}^N \rho_m u_m \bm l_m, \ \sum_{m=1}^N \frac{\rho_m(u_m^2+\sigma^2)}{2} \right)^T - \bm \rho.
        \]
        Here $\rho_m,u_m,\sigma$ are defined in Eq.~(\ref{eq:stand_rel}).

        \item In each iteration, the learning rate $l_r \in (0,1)$ is backtracked based on the Armijo-Goldstein condition \cite{arm1966} and the constraint that $\alpha_{D+1}^{(n+1)} < 0$.

        \item The recovered fluid quantity from $\bm \alpha$, i.e., $\langle \bm m_{\mathcal{L}} \exp \left( \bm \alpha \cdot \bm m_{\mathcal{L}} \right) \rangle_{\mathcal{L}}$, is slightly different from $\bm \rho$ due to the round-off error.
        To avoid error accumulation, we re-assign the results $\rho_m^{(n+1)} \gets \rho_m^{(n+1)} \left( \rho / \sum_m \rho_m^{(n+1)} \right)$, or equivalently, $\alpha_0^{(n+1)} \gets \alpha_0^{(n+1)} + \ln \left( \rho / \sum_m \rho_m^{(n+1)} \right)$.
        Here $\alpha_0$ is the first component of $\bm \alpha$.
    \end{enumerate}

    \item \textbf{Until}: $\| \nabla_{\bm \alpha} J(\bm \alpha^{(n)}) \|_2 < \epsilon$.
\end{enumerate}

In this work we take $\epsilon=10^{-8}$.
The performance of the algorithm is illustrated in Fig.~\ref{fig:optm} with different numbers of velocity direction $N$.
It is shown that the stopping criterion can be achieved with no more than 160 steps of iteration, almost regardless of $N$ from 5 to 500.
Thus, the computation of discrete equilibrium in DVDM is numerically efficient.
\begin{figure}[hbp]
    \centerline{\includegraphics[width=2.2in]{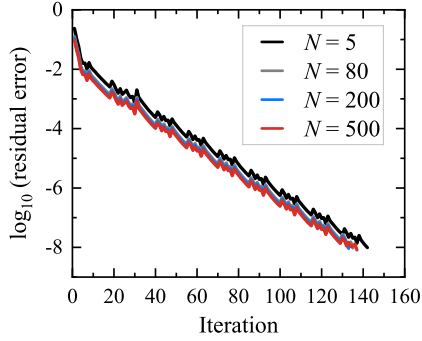}}
    \vspace*{8pt}
    \caption{Performance of the $\mathcal{E_L}$-algorithm for a given $\bm \rho \in \mathbb{R}^4$: $\rho=1, \ \bm U=(0.4,0.8)^T \text{ and } \theta=1$ with different numbers of velocity directions $N$ from 5 to 500. The directions are chosen to be $\bm l_m = (\cos \gamma_m, \sin \gamma_m)$ with $\gamma_m = \frac{m-1}{N} \pi$ for $m=1,\dots,N$.}
    \label{fig:optm}
\end{figure}

Solving the discrete equilibrium $\mathcal{E_L}$ is necessary for both DVD-DVM-II and DVD-EQMOM.
Nevertheless, for DVD-DVM-I, the discrete equilibrium is solved by minimizing $J(\bm \alpha)$ in Eq.~(\ref{eq:dvddvm_J}).
The difference between DVD-DVM-I and II is studied quantitatively in Fig.~\ref{fig:dvddvm}.
We see that for a given $\bm \rho \in \mathbb{R}^4$ and a fixed set of directions, DVD-DVM-I generally results in smaller entropy than II, but the difference almost disappears when the number of velocity nodes in each direction $N_m'>10$ while keeping the node distance $\Delta \xi$ unchanged (indicating that the nodes cover a greater range).
Thus, both DVD-DVM methods behave similarly in this case.

\begin{figure}[hbp]
    \centerline{\includegraphics[width=2.5in]{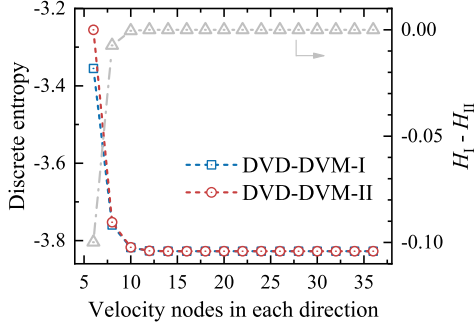}}
    \vspace*{8pt}
    \caption{Discrete entropies of DVD-DVM-I \& II (left y-axis) and their difference (right y-axis) for a given $\bm \rho \in \mathbb{R}^4$: $\rho=1$, $\bm U=(0.6,1.0)^T$, and $\theta=1.5$. The number of directions $N=8$, and the directions are chosen to be $\bm l_m = (\cos \gamma_m, \sin \gamma_m)$ with $\gamma_m = \frac{m-1}{8} \pi$ for $m=1,\dots,8$.
        The number of velocity nodes in each direction, $N_m'$ as abscissa, ranges from 6 to 36. The velocity nodes are $\xi_{mk} = \pm (k-0.5)$ for $k=1,\dots,N_m'/2$.}
    \label{fig:dvddvm}
\end{figure}

\begin{remark} [Temperature influence]
    For $D=2$, Eq.~(\ref{eq:real_th}) specifies a lower bound for the realizable temperature $\theta$. It is found that temperature has a considerable influence on the number of iterations. As revealed in Table~\ref{tab:t}, for given $\rho$ and $\bm U$, lower values of $\theta$ leads to substantially more iteration and longer CPU time.
\end{remark}

\begin{table}[ht]
    \centering
    \caption{CPU time for the $J(\bm \alpha)$-algorithm with $\rho=1$, $\bm U=(0.8,0.4)^T$ and different $\theta$}
    \begin{threeparttable}
    \begin{tabular}{c|cccccccccccc}
        \hline
        $\theta$ & 0.1 & 0.2 & 0.3 & 0.4 & 0.5 & 0.6 & 0.7 & 0.8 & 0.9 & 1.0 & 1.1 & 1.2 \\ \hline
        Time (ms) & 14.7 & 6.32 & 4.31 & 2.41 & 2.15 & 1.91 & 1.59 & 1.82 & 1.45 & 1.33 & 1.14 & 1.20 \\ \hline
    \end{tabular}
    \begin{tablenotes}
        \footnotesize
        \item The CPU time was recorded by running Matlab code on Intel(R) Core(TM) i7-1065G7. Here, 15 directions are chosen to be $\bm l_m = (\cos \gamma_m, \sin \gamma_m)$ with $\gamma_m=\frac{m-1}{15}\pi$ ($m=1,\dots,15$) and, according to Eq.~(\ref{eq:real_th}), the lower bound of $\theta$ is 0.0036.
    \end{tablenotes}
    \end{threeparttable}
     \label{tab:t}
\end{table}

\section{Numerical experiments}
\label{Sec:numsim}
Numerical results for two benchmark flow problems, the Couette flow and the 1-D Riemann problem, are reported in this section to show the performance of the DVDM-BGK models. We assume the velocity space to be two-dimensional, namely, $D=2$, $\bm \xi=(\xi_x, \xi_y)^T$ and $f=f(t,\bm x,\xi_x,\xi_y)$. Both DVD-EQMOM (Section \ref{subsec:dvd_eqmom}) and DVD-DVM (Section \ref{subsec:dvddvm}) are tested, together with the first-order schemes detailed in Section \ref{sec:numer}.

\subsection{Couette flow}

The first problem is the planar Couette flow between two infinite parallel walls located at $x=\pm H$ with a distance $L = 2H$. The left and right walls move with constant velocities $\pm v_w \bm e_y$. We set $H=0.5$ and $v_w=0.1$. The two walls drive the fluid between them from rest to a final steady state. This implies that the distribution $f$ is only dependent on $x$. We take the wall temperature $\theta_w=1$, and the Mach number $Ma=\frac{v_w}{\sqrt{\frac{D+2}{D}\theta}}=\frac{1}{10\sqrt{2}}\approx 0.0707$, representing a micro-Couette flow.

For the DVDM simulation, $N$ velocity directions are chosen to be $\bm l_m = (\cos \gamma_m, \sin \gamma_m)$ with $\gamma_m=\frac{m-1}{N}\pi$ ($m=1,...,N$). The initial temperature and velocity are taken to be 1 and $\bm 0$, respectively.
The 1-D computational domain $-0.5 < x < 0.5$ is discretized into 200 uniform cells with $\Delta x=0.005$. The time step ensures that the CFL number is less than 0.5.
While the boundary velocity and temperature are kept constant, the density at the boundaries are adjusted at every moment to ensure zero mass fluxes at the boundary.

A number of simulations for different flow regimes $\kappa = 0.1, \ 1, \text{ and } 10$ are conducted, where $\kappa = (\sqrt{\pi}/2)\rm Kn$ and the Knudsen number Kn is defined as \cite{Guo2013}
\[
{\rm Kn} = \frac{\lambda}{L} = \frac{\tau}{L}\sqrt{\frac{\pi \theta}{2}}.
\]
Here $\lambda$ denotes the mean free path of the molecules. Thus, varying $\tau$ in the source term of BGK equation gives the designed values of $\kappa$.

Fig.~\ref{fig:couette_dvdeqmom} shows the dimensionless velocity profiles for various $\kappa$ using two-node DVD-EQMOM (namely, $N_m'=2$ in Eq.~(\ref{eq:eqmom_assump})). Only the right side ($x>0$) is plotted due to symmetry. Increasing $\kappa$ from 0.1 to 10 (namely, increasing $\tau$) shifts the flow from the hydrodynamic regime (with velocity slip) to a free-molecular regime, and thus the flow is less efficient in following the moving boundaries. This property is successfully captured by all involved numbers of directions $N=5,\ 15, \text{ and } 45$.
For $N=5$, the DVD-EQMOM velocities for $\kappa=1.0 \text{ and } 10$ deviate from the DSMC results to a greater extent than that for $\kappa=0.1$. Increasing $N$ is effective in reducing the numerical errors especially for $\kappa=1.0 \text{ and } 10$.

\begin{figure}[hbp]
    \centerline{\includegraphics[width=3.8in]{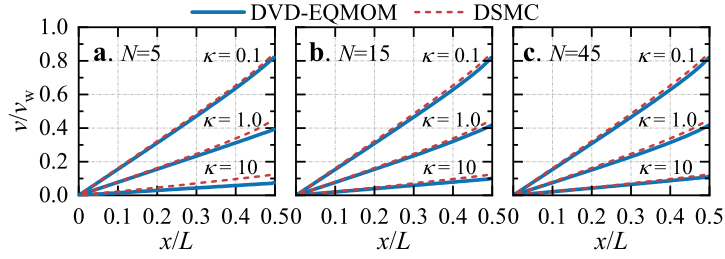}}
    \vspace*{8pt}
    \caption{Velocity profiles of the Couette flow for various Knudsen numbers $\kappa = (\sqrt{\pi}/2){\rm Kn}$ obtained with the two-node DVD-EQMOM. The number of velocity direction is varied as (a) $N=5$, (b) $N=15$, and (c) $N=45$. The directions are chosen to be $\bm l_m = (\cos \gamma_m, \sin \gamma_m)$ with $\gamma_m = \frac{m-1}{N}\pi$ for $m=1,\dots,N$. The DSMC results are extracted from Ref.~\protect\cite{Guo2013}.}
    \label{fig:couette_dvdeqmom}
\end{figure}

Fig.~\ref{fig:couette_dvddvm} presents the dimensionless velocity profiles for various $\kappa$ using both DVD-DVM-I and DVD-DVM-II with 15 prescribed velocity directions. In each direction 26 discrete velocities are selected as $\xi_{mk}=\pm(0.2k-0.1)$ for $k=1,\dots,13$. It is seen that both approaches yield very accurate velocity profiles as compared with the DSMC results. In particular, the nonlinearity of the velocity profiles near the wall is correctly reproduced.

\begin{figure}[hbp]
    \centerline{\includegraphics[width=2.54in]{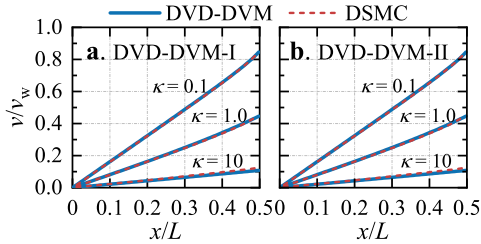}}
    \vspace*{8pt}
    \caption{Velocity profiles of the Couette flow for various Knudsen numbers $\kappa = (\sqrt{\pi}/2){\rm Kn}$ obtained with (a) DVD-DVM-I and (b) DVD-DVM-II. The number of directions $N=15$, and the directions are chosen to be $\bm l_m = (\cos \gamma_m, \sin \gamma_m)$ with $\gamma_m = \frac{m-1}{15}\pi$ for $m=1,\dots,15$. The velocity nodes in each direction are $\xi_{mk} = \pm (0.2k-0.1)$ for $k=1,\dots,13$. The DSMC results are extracted from Ref.~\protect\cite{Guo2013}.}
    \label{fig:couette_dvddvm}
\end{figure}

Fig.~\ref{fig:coshear} further shows the shear stress
$$\tau_{xy}= \langle (\xi \cos \gamma_{\mathcal{L}} - u) (\xi \sin \gamma_{\mathcal{L}} - v) f_{\mathcal{L}} \rangle_{\mathcal{L}},\quad (u,v)^T = \bm U,$$
normalized by the free-molecular stress $\tau_{\infty}$ for a variety of Knudsen numbers. The results (dots) are obtained with the two-node DVD-EQMOM and DVD-DVM-II with 15 prescribed velocity directions. The shear stress increases with Kn. The two-node DVD-EQMOM seems to overestimate the shear stress for a wide range of Kn; By contrast, the DVD-DVM-II solutions agree well with the DSMC results in the whole flow regimes.

\begin{figure}[hbp]
    \centerline{\includegraphics[width=2.2in]{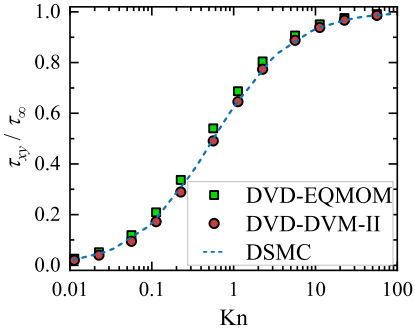}}
    \vspace*{8pt}
    \caption{Normalized stress of the Couette flow at different Knudsen numbers. The two-node DVD-EQMOM and DVD-DVM-II are both solved with 15 discrete directions $\bm l_m = (\cos \gamma_m, \sin \gamma_m)$ and $\gamma_m = \frac{m-1}{15}\pi$ for $m=1,\dots,15$. For the DVD-DVM-II, the velocity nodes in each direction are $\xi_{mk}=\pm (0.2k-0.1)$ for $k=1,\dots,13$. The DSMC results are extracted from Ref.~\protect\cite{Guo2013}.}
    \label{fig:coshear}
\end{figure}

\subsection{1-D Riemann problems}
Two 1-D Riemann problems from Refs.~\cite{Chalons2017,Fox2008} are solved in this subsection. The initial macroscopic data are
\begin{enumerate}
    \item Problem (A): $\rho(0,x)=1$ and $\theta(0,x)=\frac{1}{6}$ for all $x \in \mathbb{R}$; $U(0,x<0)=(1,0)^T$ and $U(0,x>0)=(-1,0)^T$.
    \item Problem (B): $\rho(0,x<0) = 3.093$ and $\rho(0,x>0)=1$; $U(0,x)=(0,0)^T$ and $\theta(0,x)=1$ for all $x \in \mathbb{R}$.
\end{enumerate}
The corresponding initial distribution $f=f(t,x,\xi_x,\xi_y)$ are taken to be in equilibrium determined by the macroscopic data above.

As for the collision term, two limiting cases are considered: (\textit{i}) $\tau=\infty$ and thus the collision term vanishes. In this case, the distribution $f$ is a traveling wave and its analytical expression is given in Ref.~\cite{Chalons2017}; (\textit{ii}) $\tau=0$ corresponds to the classical Euler equation for an inviscid compressible fluid.
Notice that the specific heat ratio in the analytical expression \cite{Toro2009} should be taken as 2 due to the 2-D velocity space assumption.

In the DVDM simulation, 15 velocity directions are chosen to be $\bm l_m = (\cos \gamma_m, \sin \gamma_m)$ with $\gamma_m=\frac{m-1}{15}\pi$ for $m=1,\dots,15$. For the DVD-EQMOM, we set $N_m'=2$; for the DVD-DVM-II, the velocity nodes in each direction are chosen as $\xi_{mk}=\pm(0.2k-0.1)$ for $k=1,\dots,13$. The computational domain is taken to be $-2<x<2$ and is discretized into 400 uniform cells with $\Delta x = 0.01$. The time step ensures that the CFL number is less than 0.5.
For $\tau=0$, the collisions reset the variables (either discretized distributions or moments) to be solved as the equilibrium states at the end of each time step.

Figs.~\ref{fig:1DRA_NC} \& \ref{fig:1DRA_Eu} exhibit the spatial distribution of macroscopic quantities at $t=0.2$ for Problem (A) with $\tau=\infty$ (no collision) and $\tau=0$ (ultrafast collision), respectively. For $\tau=\infty$, the DVD-DVM-II results agree reasonably well with the analytic solution, while the two-node EQMOM is less accurate in the region of $-0.5<x<0.5$. In particular, the density at $x=0$ is underestimated, and the simulated energy (or temperature) at $x=0$ is greater than the true value.
For $\tau=0$, both models well capture the shock waves propagating towards both sides, but there are some inaccuracies in the central region around $x=0$ (the contact surface). The inaccuracies may be attributed to either the DVDM assumptions or numerical schemes. It is interesting to remark that such a deviation around $x=0$ was also observed in the two-node EQMOM solutions in Refs.~\cite{Chalons2017,HLY2020}.

\begin{figure}[hbp]
    \centerline{\includegraphics[width=3.85in]{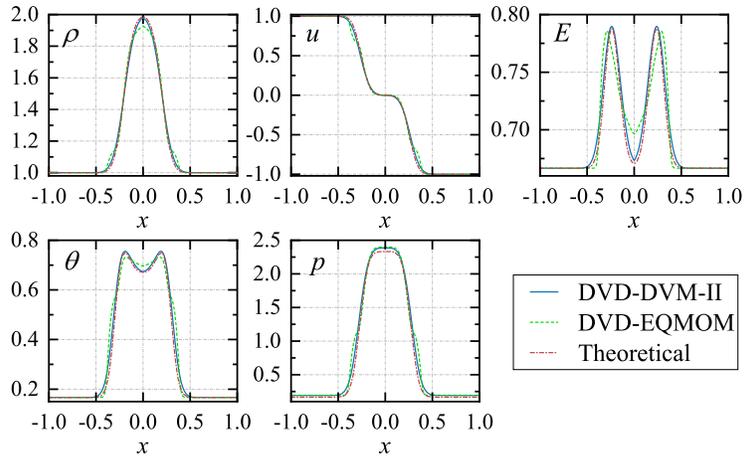}}
    \vspace*{8pt}
    \caption{Macroscopic quantities at $t=0.2$ for the Riemann problem (A) with $\tau=\infty$ (i.e., no collision).}
    \label{fig:1DRA_NC}
\end{figure}

\begin{figure}[hbp]
    \centerline{\includegraphics[width=3.85in]{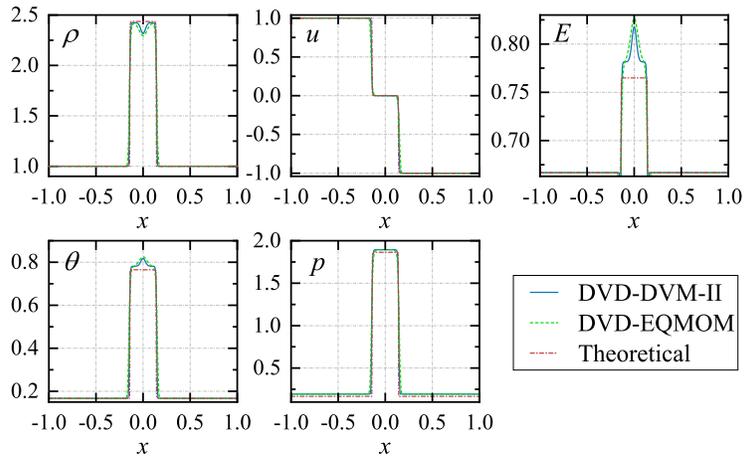}}
    \vspace*{8pt}
    \caption{Macroscopic quantities at $t=0.2$ for the Riemann problem (A) with $\tau=0$ (i.e., ultrafast collision).}
    \label{fig:1DRA_Eu}
\end{figure}

Figs.~\ref{fig:1DRB_NC} \& \ref{fig:1DRB_Eu} illustrate the spatial distribution of macroscopic quantities at $t=0.2$ for Problem (B) with $\tau=\infty$ (no collision) and $\tau=0$ (ultrafast collision), respectively. As in Problem (A), for $\tau=\infty$ the DVD-DVM is sufficiently accurate in this free-moving situation, while the two-node EQMOM introduces non-negligible error for the energy (or temperature) profile. By contrast, in the Euler limit $\tau=0$, both models exhibit correctly the rarefaction wave, contact discontinuity and shock wave. However, the discontinuities are less sharper than the theoretical solution, especially for $\rho$ and $u$, which may be mainly caused by the first-order schemes. In comparision with the DVD-DVM, the two-node DVD-EQMOM is less accurate in describing the right-moving shock wave located at $x \approx 0.35$.

\begin{figure}[hbp]
    \centerline{\includegraphics[width=3.85in]{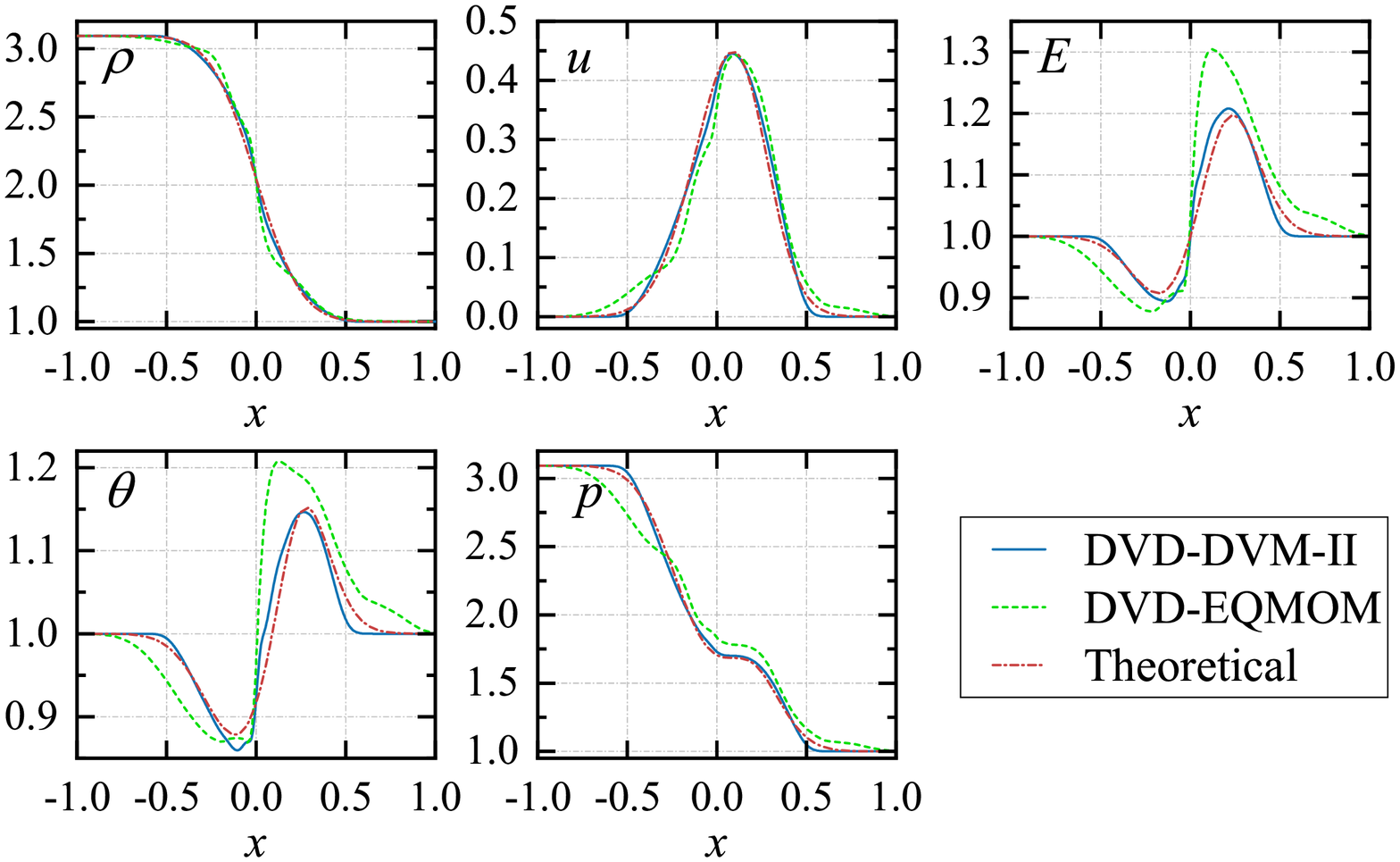}}
    \vspace*{8pt}
    \caption{Macroscopic quantities at $t=0.2$ for the Riemann problem (B) with $\tau=\infty$ (i.e., no collision).}
    \label{fig:1DRB_NC}
\end{figure}

\begin{figure}[hbp]
    \centerline{\includegraphics[width=3.85in]{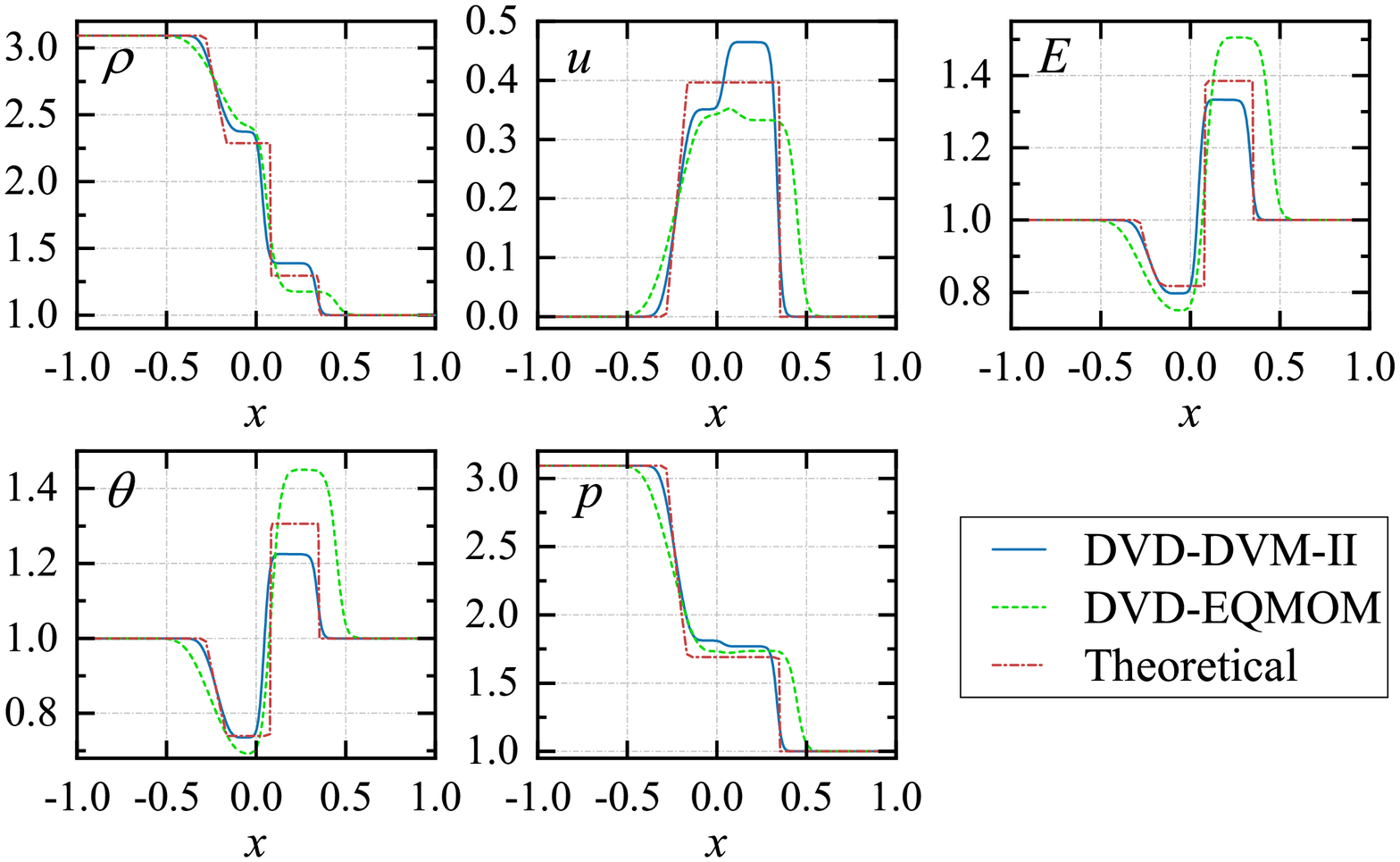}}
    \vspace*{8pt}
    \caption{Macroscopic quantities at $t=0.2$ for the Riemann problem (B) with $\tau=0$ (i.e., ultrafast collision).}
    \label{fig:1DRB_Eu}
\end{figure}

\begin{remark}
    In solving Problem (B) with $\tau=0$ and 15 directions (results given in Fig.~\ref{fig:1DRB_Eu}), the CPU times of the 2-node DVD-EQMOM and DVD-DVM-II (26 nodes in each direction) are 120 s and 500 s, respectively. The CPU times were recorded by running Matlab codes on Intel(R) Core(TM) i7-1065G7.
\end{remark}

\section{Conclusions and Perspectives}
\label{Sec:concl}

In this paper, we present a novel discrete-velocity-direction model (DVDM) with a minimum entropy principle. As a semi-continuous discretization model of BGK-type, it assumes that the molecule velocity has a few prescribed directions, but the velocity modulus is still continuous. The Maxwellian equilibrium is defined as a minimizer of a discrete entropy subject to conservation laws of density, momentum and energy.
We show that the discrete equilibrium is uniquely determined by macroscopic fluid quantities computed with nonnegative distributions generating finite density and energy, implying that the DVDM is well defined.
Numerically, the discrete equilibrium with a given macroscopic flow state can be computed efficiently by convex optimization algorithms. Moreover, the model ensures positivity of the solutions and has a proper version of $H$-theorem.

The proposed model has several advantages. Firstly, it provides a new way in choosing discrete velocities for the computational practice of the conventional discrete-velocity methodology.
In this way, we replace the continuous one-dimensional velocity moduli of DVDM with finite nodes in each direction and derive new discrete-velocity models (called DVD-DVM).
Furthermore, the 1-D EQMOM can be implemented in each prescribed discrete direction to derive multidimensional hyperbolic EQMOM conveniently.
As reported in the literature \cite{Chalons2017,MF2013}, 1-D EQMOM is quite successful in a number of applications, and it also has good mathematical properties \cite{HLY2020}.
However, a satisfactory multidimensional version of EQMOM is not available for a long time.
In this sense, we have offered a solution to the multidimensional extension of the 1-D EQMOM.
To show the performance of the above models, we simulate two benchmark flows, the Couette flow and 1-D Riemann problems, with reasonable results at moderate computational costs.

Potential focuses of future work can be directed to: (\textit{i}) studying the DVDM mathematically, including the realizable condition and the structural stability condition \cite{HLY2020,Yong1999}; (\textit{ii}) simulating real-world flows with 3-D velocity space and more complex boundary conditions with higher-order numerical schemes.

\section*{Acknowledgment}
This work is supported by the National Natural Science Foundation
of China (Grant no. 51906122 and 12071246) and National Key Research and Development Program of China (Grant no. 2021YFA0719200). The authors acknowledge Mr. Jialiang Zhou, Mr. Qiqi Rao and Prof. Shuiqing Li at Tsinghua University for valuable discussions.

\appendix

\section{Realizabily of the DVDM}
\label{sec:appendix_realize}
In this appendix, we show that for the two-dimensional DVDM, more macroscopic states can be realized with more discrete-velocity directions.
In this case, the discrete-velocity directions can be expressed as $\bm l_m = (\cos \gamma_m, \sin \gamma_m)^T$ with $\gamma_m \in [0,\pi)$ and therefore can be understood as the complex numbers $e^{i \gamma_m}$. Without loss of generality, we assume $0 \le \gamma_1 < \dots < \gamma_N < \pi$. The main result here is the following lemma.

\begin{lemma} \label{lem:2dSumMin}
    For any $N$-tuple $(a_m)_{m=1}^N \in \mathbb{R}$ such that $$\sum_{m=1}^N a_m e^{i \gamma_m} = 1,$$
    we have
    \begin{equation} \label{eq:minam}
        \sum_{m=1}^N |a_m| \ge \frac{\sin \gamma_1 + \sin \gamma_N}{\sin (\gamma_N - \gamma_1)}.
    \end{equation}
    The minimum is attained when
    \[
    a_1=\frac{\sin \gamma_N}{\sin (\gamma_N - \gamma_1)}, \ a_N=\frac{-\sin \gamma_1}{\sin (\gamma_N - \gamma_1)}, \text{ and } a_2=\dots=a_{N-1}=0.
    \]
\end{lemma}

\begin{proof}
    Set $z_m = a_m e^{i \gamma_m}$. According to the assumption, we have $\sum_m z_m = 1$ and
    $$\sum_m |a_m| = \sum_m |z_m| \ge \left \lvert \sum_m z_m \right \rvert = 1.$$
    Therefore the lemma obviously holds with $\gamma_1=0$.

    For $\gamma_1>0$, since $\gamma_1 \leq \gamma_m<\pi$ for all $m$, $z_m$ is not a real number unless $z_m=0$. Thus, from $\sum_m z_m=1$ it follows that there must be some $z_m$ with positive imaginary part (denote their sum by $z^+$) and some $z_m$ with negative imaginary part (denote their sum by $z^-$). Obviously, we have $\gamma_+ := \arg z^+ \in (0,\pi)$ and $\gamma_- := \arg z^- \in (-\pi,0)$. Note that $z^+ + z^- = 1$, meaning that the three complex numbers $z^+$, $z^-$ and $1$ constitute a triangle (after a proper shift of $z^-$) in the complex plane. Then the angle $(\gamma_+ - \gamma_-)$ from $z^-$ to $z^+$ is in $(0, \pi)$. We thus see from the law of sines:
    \[
    \frac{1}{\sin(\gamma_+ - \gamma_-)} = \frac{|z^+|}{\sin (- \gamma_-)} = \frac{|z^-|}{\sin \gamma_+}
    \]
    that
    \[
    \sum_m |a_m| = \sum_m |z_m| \ge |z^+| + |z^-| = \frac{\sin \gamma_+ - \sin \gamma_-}{\sin (\gamma_+ - \gamma_-)} =: S(\gamma_+, \gamma_-),
    \]
    which is monotonically increasing with $\gamma_+$ and decreasing with $\gamma_-$. This can be seen by computing the derivatives:
    \[
    \begin{aligned}
        \frac{\partial S}{\partial \gamma_+}
        & = \frac{\cos \gamma_+ \sin(\gamma_+ - \gamma_-) - (\sin \gamma_+ - \sin \gamma_-) \cos(\gamma_+ - \gamma_-)}{\sin^2(\gamma_+ - \gamma_-)} \\
        & = \frac{\sin \gamma_- (\cos(\gamma_+ - \gamma_-) - 1)}{\sin^2(\gamma_+ - \gamma_-)}
        \geq 0,
    \end{aligned}
    \]
    and similarly $\frac{\partial S}{\partial \gamma_-} \leq 0$. Geometrically we can easily see that $\gamma_+ \ge \gamma_1$ and $\gamma_- \le \gamma_N - \pi$.
    Thus, the minimum of $|z^+| + |z^-|$ is attained when $\gamma_+ = \gamma_1$ and $\gamma_- = \gamma_N - \pi$.
    This corresponds to $a_2=\dots=a_{N-1}=0$ and leads exactly to the RHS of Eq.~(\ref{eq:minam}).
\end{proof}

With this lemma, we can explicitly write down the constraint on temperature $\theta$. To do this, we denote by $\gamma_{\bm U} \in (-\pi,\pi]$ the argument of the flow velocity $\bm U$ in Eq.~(\ref{eq:relaxed_U}). Then the second equality in Eq.~(\ref{eq:relaxed_U}) can be rewritten as
\[
\sum_{m=1}^N \hat\rho_m \hat u_m e^{i \gamma_m} = \rho U e^{i \gamma_{\bm U}}
\]
or
\begin{equation}\label{eq:A2}
    \sum_{m=1}^N s\hat\rho_m \hat u_m e^{i (\gamma_m - \eta)} = \rho U,
\end{equation}
where
\[
s = \left \{
\begin{aligned}
    1, \quad &\text{if } 0 < \gamma_{\bm U} \le \pi, \\
    -1, \quad &\text{if } -\pi < \gamma_{\bm U} \le 0,
\end{aligned}
\right.
\]
and $\eta = \gamma_{\bm U} - \frac{s-1}{2}\pi \in (0,\pi]$.

Let $k$ ($1\le k \le N$) be such that
\begin{equation*}
    \gamma_k < \eta \le \gamma_{k+1}.
\end{equation*}
For convenience, we define $\gamma_{N+1}=\pi + \gamma_1$.
Then we set
\[
\tilde \gamma_m =
\left\{
\begin{aligned}
    & \gamma_m + \pi - \eta,& 1\leq m\leq k, \\
    & \gamma_m - \eta, & k+1\leq m \leq N.
\end{aligned}
\right.
\]
It is clear that
$0\leq \tilde\gamma_{k+1}<\cdots < \tilde \gamma_N < \tilde \gamma_1 < \cdots < \tilde \gamma_k < \pi$. Now we apply the lemma above to Eq.~(\ref{eq:A2}) to get
$$
U^2 \left( \frac{\sin \tilde \gamma_{k+1} + \sin \tilde \gamma_k}{\sin (\tilde \gamma_k - \tilde \gamma_{k+1})} \right)^2
\le
\frac{1}{\rho^2}\left( \sum_{m=1}^N |\hat\rho_m \hat u_m| \right)^2
<
2 E,$$
where the last inequality is Eq.~(\ref{eq:3.2}). The left-hand side
is
\[
U^2\left( \frac{\sin(\gamma_{k+1}-\eta) + \sin(\eta - \gamma_k)}{\sin (\gamma_{k+1} - \gamma_k)} \right)^2,
\]
which obviously goes to $U^2$ as $\gamma_{k+1}-\gamma_k \rightarrow 0$. Moreover, the constraint for $\theta = (2E - U^2)/2$ becomes
\begin{equation}\label{eq:real_th}
    \theta > U^2 \frac{\sin(\gamma_{k+1}-\eta) \sin(\eta - \gamma_k)}{1+ \cos(\gamma_{k+1} - \gamma_k)}.
\end{equation}
This lower bound approaches 0 when $\gamma_{k+1} - \gamma_k \to 0$.

\end{document}